\documentclass[aps,prl,reprint,superscriptaddress,longbibliography]{revtex4-2}

\usepackage{graphicx}
\graphicspath{{./}{./figure/}}

\usepackage{amsmath,amssymb,bm,braket}
\usepackage{enumitem}

\usepackage{hyperref}
\hypersetup{colorlinks=true, citecolor=magenta, linkcolor=blue, urlcolor=cyan}

\usepackage{amsthm}
\newtheorem{theorem}{Theorem}
\newtheorem{lemma}{Lemma}
\newtheorem{proposition}{Proposition}
\newtheorem{theoremNonint}{Theorem}

\makeatletter
\newcommand{\fmarki}{}
\newcommand{\fmarkii}{\ensuremath{\dagger}}
\newcommand{\fmarkiii}{\ensuremath{\ddagger}}
\def\@fnsymbol#1{{\ifcase#1\or \fmarki\or \fmarkii\or \fmarkiii \else\@ctrerr\fi}}
\makeatother

\begin{document}

\title{Exact Thermal Eigenstates of Nonintegrable Spin Chains at Infinite Temperature}

\thanks{
    \hspace{-0.88em}%
    \normalfont\textsuperscript{*}%
    \hspace{+0.43em}%
    These authors contributed equally to this work.
    \phantomsection\label{thanks}
}

\author{Yuuya Chiba\normalfont\textsuperscript{\hyperref[thanks]{*},}}
\email{yuya.chiba@riken.jp}
\affiliation{Nonequilibrium Quantum Statistical Mechanics RIKEN Hakubi Research Team, RIKEN Cluster for Pioneering Research (CPR), 2-1 Hirosawa, Wako, Saitama 351-0198, Japan}

\author{Yasushi Yoneta\normalfont\textsuperscript{\hyperref[thanks]{*},}}
\email{yasushi.yoneta@riken.jp}
\affiliation{Center for Quantum Computing, RIKEN, 2-1 Hirosawa, Wako, Saitama 351-0198, Japan}

\date{\today}

\begin{abstract}
The eigenstate thermalization hypothesis (ETH) plays a major role in explaining thermalization of isolated quantum many-body systems. However, there has been no proof of the ETH in realistic systems due to the difficulty in the theoretical treatment of thermal energy eigenstates of nonintegrable systems. Here, we write down analytically
thermal eigenstates of nonintegrable spin chains. We consider a class of theoretically tractable volume-law states, which we call entangled antipodal pair (EAP) states. These states are thermal, in the most fundamental sense
that they are indistinguishable from the Gibbs state with respect to all local observables, with infinite temperature. We then identify Hamiltonians having the EAP state as an eigenstate and rigorously show that some of these Hamiltonians are nonintegrable. Furthermore, a thermal pure state at an arbitrary temperature is obtained by the imaginary-time evolution of an EAP state. Our results offer a potential avenue for providing a provable example of the ETH.
\end{abstract}

\maketitle

\textit{Introduction---}
Understanding the mechanism of thermalization in quantum many-body systems has been a pivotal issue in statistical physics~\cite{DAlessio2016,Gogolin2016,Mori2018}. Notably, the eigenstate thermalization hypothesis (ETH)~\cite{Neumann1929,Deutsch1991,Srednicki1994} has served as a cornerstone in this field. It posits that all the energy eigenstates in the bulk of the spectrum of quantum many-body systems exhibit thermal properties, thereby giving a plausible explanation of thermalization.

While the ETH is anticipated to hold in most nonintegrable systems, the verification of whether this hypothesis holds in realistic many-body systems relies on numerical calculations, and a theoretical verification has remained elusive~\cite{Rigol2008,Kim2014,Beugeling2014,Steinigeweg2014}. Thus, a significant challenge lies in theoretically addressing the nature of energy eigenstates, particularly in nonintegrable systems. However, it has not been clear whether thermal eigenstates of nonintegrable systems can be treated theoretically. This stems from the difficulty of writing down quantum states whose entanglement entropy obeys a volume law.

One approach to treat quantum many-body states theoretically is to use variational wave functions. Particularly for states that contain a small amount of entanglement, they can be represented via tensor network states such as a matrix product state~\cite{Verstraete2008,Schollwock2011}. Tensor network states are highly tractable, making them not only practical but also significantly contributing to theoretical advancements. Indeed, by utilizing the matrix product state, it has been successful to exactly describe finite-energy-density low-entangled (thus nonthermal) eigenstates even for nonintegrable systems~\cite{Moudgalya2018_1,Moudgalya2018_2,Moudgalya2020}, which are examples of many-body scars~\cite{Shiraishi2017,Mori2017,Bernien2017,Turner2018a,Turner2018b}. However, there has been a lack of variational wave functions suitable for theoretical analysis of volume-law states, which is one of the reasons why thermal eigenstates have not yet been obtained. Hence, there is a craving for a class of volume-law states amenable to the theoretical treatment~\cite{Vitagliano2010,Ramirez2014,Ramirez2015,Langlett2022,Bettaque2024}.

In this Letter, we provide
pairs of a nonintegrable Hamiltonian and its thermal eigenstate at infinite temperature.
We consider a class of volume-law states, which we call the entangled antipodal pair (EAP) states, that are amenable to theoretical calculations.
Then we fully characterize Hamiltonians having the EAP state as an eigenstate.
It is rigorously shown that some of these Hamiltonians are nonintegrable.
In addition, by evolving an EAP state in imaginary time, we construct a thermal pure state at arbitrary temperature, which is locally indistinguishable from the Gibbs state.

\begin{figure}
    \centering
    \includegraphics[width=\linewidth]{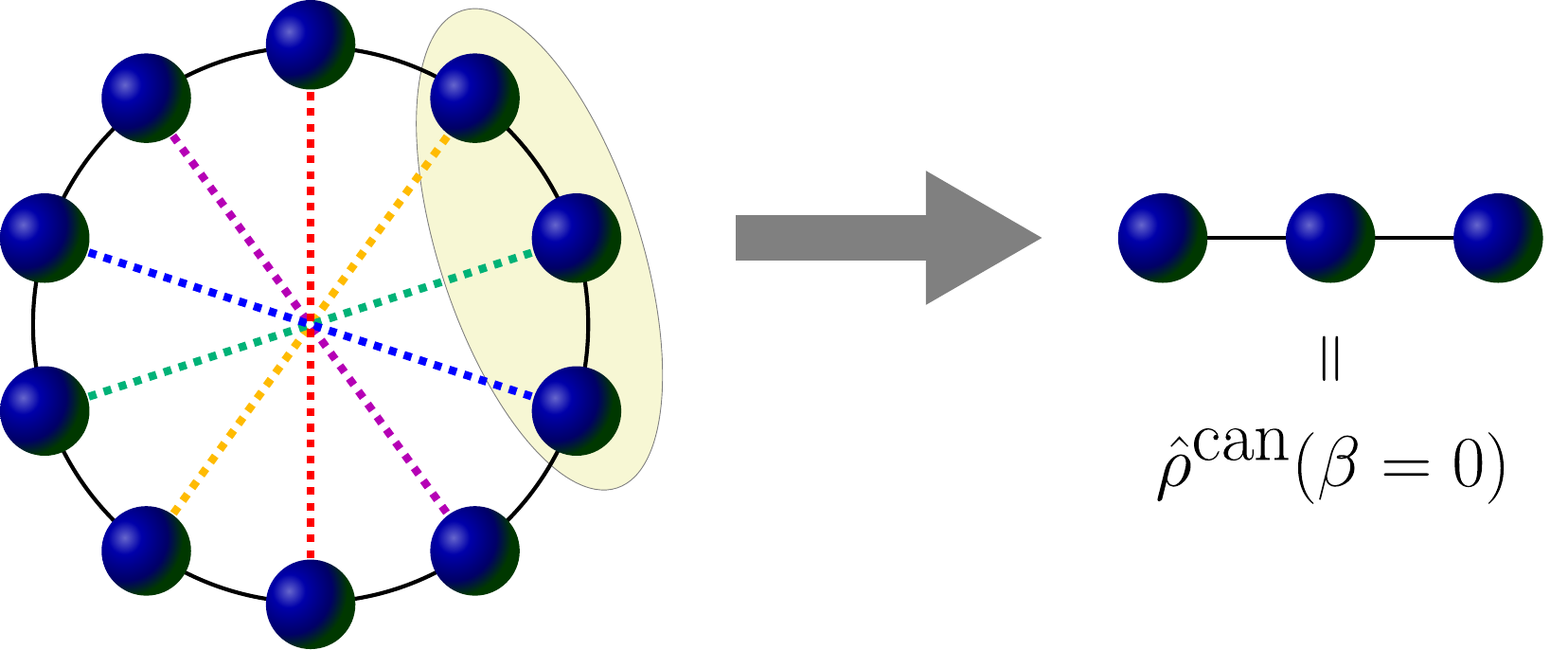}
    \caption{Schematic diagram depicting an entangled antipodal pair state $\ket{\mathrm{EAP}}$. The antipodal pairs of spins linked by dotted lines are in the Bell states $\ket{\Phi_{pq}}$. For any subsystem with diameter smaller than or equal to half of the size of the entire system $N$, the reduced density matrix coincides with the maximally mixed state, i.e., the Gibbs state $\hat{\rho}^\mathrm{can}$ at the inverse temperature $\beta=0$.}
    \label{fig:EAP}
\end{figure}

\textit{Entangled antipodal pair state---}
We consider quantum spin-$1/2$ systems on a one-dimensional lattice $\Lambda=\{1,2,\cdots,N\}$ with periodic boundary conditions. We assume that the number of lattice sites $N$ is even. Let $\hat{\sigma}_{j}^\mu (\mu=x,y,z)$ be the Pauli matrices acting on the $j$-th site, and $\ket{0}_j$ and $\ket{1}_j$ be the eigenvectors of $\hat{\sigma}_{j}^z$.

Our goal is to obtain pairs of a nonintegrable Hamiltonian and its thermal eigenstate. To achieve this, we adopt the following strategy. First, we introduce a class of volume-law states that are theoretically tractable. Next, for each of these volume law states, we search for Hamiltonians that have it as a thermal eigenstate. Finally, we prove that some of these Hamiltonians are nonintegrable.
Following this strategy, as the first step, we introduce the {\em entangled antipodal pair (EAP) state} as~\footnote{Some special EAP states appear in the study of 
integrable systems~\cite{Caetano2022,Ekman2022}.}
\begin{align}
    \ket{\mathrm{EAP}} = \bigotimes_{j=1}^{N/2} \ket{\Phi_{p_jq_j}}_{j,j+N/2}.
\end{align}
Here $\ket{\Phi_{p_jq_j}}_{j,j+N/2} (p_j,q_j=0,1)$ are the Bell states between antipodal sites $j$ and $j+N/2$~\footnote{It is possible to extend EAP states to higher-dimensional systems. However, unlike the one-dimensional case, there exists arbitrariness in the choice of antipodal sites.} defined by
\begin{align}
    &\ket{\Phi_{pq}}_{j,j'} = \frac{\ket{0}_{j}\ket{p}_{j'}+(-1)^q\ket{1}_{j}\ket{{\bar{p}}}_{j'}}{\sqrt{2}} \quad (p,q=0,1),
\end{align}
where $\bar{p}$ represents the negation of $p$.
It is straightforward to check that
\begin{align}
    \hat{\sigma}_{j}^\mu \hat{\sigma}_{j+N/2}^\mu \ket{\Phi_{p_jq_j}}_{j,j+N/2}
    = \omega_j^\mu \ket{\Phi_{p_jq_j}}_{j,j+N/2},
    \label{eq:sigma-sigma-Phi}
\end{align}
where $\omega_j^z=(-1)^{p_j}$, $\omega_j^x=(-1)^{q_j}$, and $\omega_j^y = - \omega_j^x \omega_j^z$. Hence, the EAP state is uniquely characterized by $(\omega_j^x)_{j\in\Lambda}$ and $(\omega_j^z)_{j\in\Lambda}$.
As a consequence of Eq.~\eqref{eq:sigma-sigma-Phi}, the action of $\hat{\sigma}_{j}^\mu$ on the EAP state is equivalent to the action of $\hat{\sigma}_{j+N/2}^\mu$, except for the factor of $\omega_j^\mu$, i.e.,
\begin{align}
    \hat{\sigma}_{j}^\mu \ket{\mathrm{EAP}}
    = \omega_j^\mu \hat{\sigma}_{j+N/2}^\mu \ket{\mathrm{EAP}}.
    \label{eq:sigma-EAP}
\end{align}
This plays a key role in the proof of our main results.

Now we explain that EAP states are thermal.
Since an EAP state consists of Bell pairs between sites separated by $N/2$~(Fig.~\ref{fig:EAP}), for any subsystem $X$ with diameter
\begin{align}
    D(X) = \max_{j,j' \in X} |j-j'|
\end{align}
satisfying $D(X)<N/2$, the entanglement entropy obeys a volume law with a maximum coefficient of $\log 2$.
Consequently, the reduced density matrix of an EAP state for the subsystem $X$ is the maximally mixed state $\propto \hat{I}_X$, which coincides with that of the Gibbs state $\hat{\rho}^\mathrm{can} \propto e^{- \beta \hat{H}}$ at the inverse temperature $\beta=0$.
Thus, EAP states cannot be distinguished from the thermal equilibrium state at $\beta=0$ by measurements of any local observable, i.e., observable whose support size is independent of $N$. This means that EAP states represent ``microscopic thermal equilibrium'' (MITE)~\cite{Goldstein2015,Goldstein2017,Mori2018}, which is one of the most fundamental definitions of equilibrium states, e.g., in the context of thermalization.

This should be contrasted with the rainbow state~\cite{Vitagliano2010,Ramirez2014,Ramirez2015,Langlett2022}, which is a product of the Bell states between sites $j$ and $N-j+1$. 
If we focus only on a subsystem $\{1,2,\cdots,\ell\}$ (with $\ell\le N/2$), the reduced density matrix of a rainbow state is maximally mixed, and hence
coincides with that of the Gibbs state $\hat{\rho}^\mathrm{can} \propto e^{- \beta \hat{H}}$ at $\beta=0$. However, for instance, by considering a local observable $\hat{\sigma}^{x}_{N/2}\hat{\sigma}^{x}_{N/2+1}$, rainbow states can be distinguished from the Gibbs state (because the expectation value of $\hat{\sigma}^{x}_{N/2}\hat{\sigma}^{x}_{N/2+1}$ in a rainbow state takes $\pm 1$ while that in the Gibbs state takes $0$). Thus, rainbow states do not represent MITE.

Note that, although EAP states represent MITE, they are very different from typical thermal eigenstates of nonintegrable systems. 
For instance, 
the expectation value of a few-body but nonlocal observable $\hat{\sigma}_{j}^x \hat{\sigma}_{j+N/2}^x$ in an EAP state 
differs from that in the Gibbs state $\hat{\rho}^{\mathrm{can}}$ at $\beta=0$. (The former takes $\pm 1$, while the latter takes $0$.)
In other words, EAP states are examples of states representing MITE but not ``few-body thermal equilibrium~\cite{Mori2018}.''
Although MITE is more common, some studies~\cite{Rigol2008,Rigol2009,Rigol2009a,Santos2010,Beugeling2014,Sugimoto2023} have also examined few-body thermal equilibrium and shown that thermal eigenstates of nonintegrable systems often represent it. These facts indicate that EAP states are sort of special thermal eigenstates~\footnote{ In addition, we can show that EAP states are not thermal in a ``deep sense''~\cite{Ippoliti2023} (even when the moment $k$ of the projected ensemble is not so large, such as $k=2$) in contrast to typical eigenstates of nonintegrable systems~\cite{Cotler2023}.}.

\textit{EAP state as an energy eigenstate---}
As explained above, in the following, we search for Hamiltonians that have the EAP state as a thermal eigenstate. To this end, we write the Hamiltonian in the most general form as
\begin{align}
    \hat{H}
    &= \sum_{X(\subset\Lambda)} \sum_{\vec{\mu}\in\{x,y,z\}^{X}}
    J_{X}^{\vec{\mu}} \bigotimes_{j\in X} \hat{\sigma}_{j}^{\mu_{j}}.
    \label{eq:Hamiltonian}
\end{align}
In the first sum, $X$ represents a subset of $\Lambda$. In the second sum, $\vec{\mu}$ represents a combination of $\mu_{j}=x,y,z$ for $j \in X$. Here the origin of the energy is taken such that $\hat{H}$ is traceless.
We are interested in the Hamiltonian where some EAP state $\ket{\mathrm{EAP}}$ is an energy eigenstate. 
Imposing a very mild condition that the locality of interactions in $\hat{H}$ is less than $N/4$, we can characterize such Hamiltonians completely as follows~\footnote{See Supplemental Material for proofs of theorems, additional numerical results, and discussions, which includes Refs.~\cite{Shiraishi2019,Park2024,Mehta2004,Atas2013,Araki1969,Araki1975}}:
\begin{theorem} \label{theorem:eigenstate-condition}
Suppose that coefficients $J_{X}^{\vec{\mu}}$ defined in Eq.~\eqref{eq:Hamiltonian} are zero for all subsets $X$ with $D(X)\ge N/4$. Then the following three statements are equivalent:
\begin{enumerate}[label={(\roman*)},ref={\roman*}]
    \item \label{statement:eigenstate-condition_1}
        An EAP state $\ket{\mathrm{EAP}}$ is an eigenstate of $\hat{H}$.
    \item \label{statement:eigenstate-condition_2}
        An EAP state $\ket{\mathrm{EAP}}$ is an eigenstate of $\hat{H}$ with the eigenvalue $0$.
    \item \label{statement:eigenstate-condition_3}
        For all $X \subset \Lambda$ and $\vec{\mu} \in \{x,y,z\}^{X}$, 
        \begin{align}
            J_{Y}^{\vec{\nu}}
            = - J_{X}^{\vec{\mu}} \prod_{j\in X} \omega_{j}^{\mu_{j}},
            \label{eq:eigenstate-condition}
        \end{align}
        where $Y=X+N/2$ is the translation of $X$ by $N/2$ sites and $\nu_{j} = \mu_{j-N/2}$ for $j\in Y$.
\end{enumerate}
\end{theorem}

Note that we can easily check whether a given EAP state is an eigenstate of a given Hamiltonian by testing Eq.~\eqref{eq:eigenstate-condition}. This provides a simple explanation to previous results~\cite{Udupa2023,Ivanov2024} that construct thermal energy eigenstates in certain nonintegrable systems~\cite{Note4}.

Note also that Eq.~\eqref{eq:eigenstate-condition} contains only a pair of coefficients $\{J_{X}^{\vec{\mu}},J_{Y}^{\vec{\nu}}\}$, but no other coefficients. 
This means that any coefficients $J_{X_1}^{\vec{\mu}_1}$ and $J_{X_2}^{\vec{\mu}_2}$ belonging to different pairs $\{J_{X_1}^{\vec{\mu}_1},J_{Y_1}^{\vec{\nu}_1}\}$ and $\{J_{X_2}^{\vec{\mu}_2},J_{Y_2}^{\vec{\nu}_2}\}$ can be taken independently.
Thus, for any EAP state, we can construct a Hamiltonian that has the EAP state as an eigenstate by just taking the coefficients such that every pair $\{J_{X}^{\vec{\mu}},J_{Y}^{\vec{\nu}}\}$ satisfies Eq.~\eqref{eq:eigenstate-condition}. [In order to restrict the interactions in $\hat{H}$ to a finite range, we can take coefficients $J_{X}^{\vec{\mu}}$ with large $D(X)$ to be zero because $J_{X}^{\vec{\mu}}=J_{Y}^{\vec{\nu}}=0$ trivially satisfies Eq.~\eqref{eq:eigenstate-condition}.] However, it is not obvious whether such a Hamiltonian can be translation invariant, even if the EAP state is translation invariant and coefficients are appropriately chosen. Therefore, in the following, we impose translation invariance on the Hamiltonian and show that only several EAP states are allowed as solutions of Eq.~\eqref{eq:eigenstate-condition}.

\textit{Translation-invariant nonintegrable Hamiltonians---}
Suppose that $\hat{H}$ defined in Eq.~\eqref{eq:Hamiltonian} is translation invariant and consists of interactions up to nearest neighbor sites. Then, $\hat{H}$ is characterized by nine nearest-neighbor-interaction coefficients $\{J^{\mu \nu}\}_{\mu,\nu=x,y,z}$ and three magnetic field coefficients $\{h^{\mu}\}_{\mu=x,y,z}$.
Since Eq.~\eqref{eq:eigenstate-condition} reduces to 
\begin{align}
    &J^{\mu_1 \mu_2}(1+\omega_{j}^{\mu_1}\omega_{j+1}^{\mu_2}) = 0, \quad
    h^{\mu_1}(1+\omega_{j}^{\mu_1}) = 0,\nonumber\\
    &\quad \text{for all $\mu_1,\mu_2\in\{x,y,z\}$ and $j\in\Lambda$},
    \label{eq:Reduced_eigenstate-condition}
\end{align}
we can solve them and find all
possible choices of $J^{\mu_1\mu_2},h^{\mu_1}\neq 0$ and $(\omega_{j}^{x},\omega_{j}^{z},\omega_{j}^{y}=-\omega_{j}^{x}\omega_{j}^{z})_{j\in\Lambda}$, as will be described in the following Theorem~\ref{theorem:model}.
Interestingly, there are some nontrivial solutions whose $(\omega_{j}^{x},\omega_{j}^{z},\omega_{j}^{y})_{j\in\Lambda}$ are not invariant by single-site shift but invariant by $n$-site shift for $n>1$. To express such solutions efficiently, we introduce the following notation for EAP states that are invariant by $n$-site shift~\footnote{Note that, for some of states~(\ref{eq:alt-notation}), $n$-site shift can cause a phase change, such as $\hat{\mathcal{T}}^3\ket{3;01,10,11}=-\ket{3;01,10,11}$, where $\hat{\mathcal{T}}$ is the one-site translation operator.}:
when $N/2$ is a multiple of $n$,
\begin{align}
   \ket{n;p_{1}^{*}q_{1}^{*},\cdots,p_{n}^{*}q_{n}^{*}}
   \label{eq:alt-notation}
\end{align}
denotes the EAP state characterized by $(p_{j},q_{j})_{j\in\Lambda}$ satisfying $p_{mn+j}=p_{j}^{*}$ and $q_{mn+j}=q_{j}^{*}$ for $m=0,1,\cdots,N/2n-1$ and $j=1,2,\cdots,n$.
For example, when $N/2$ is even, $\ket{2;10,11}=\bigotimes_{j=1}^{N/4}\ket{\Phi_{10}}_{2j-1,2j-1+N/2}\bigotimes_{j=1}^{N/4}\ket{\Phi_{11}}_{2j,2j+N/2}$.
Using this notation, we can list all nontrivial solutions as follows~\cite{Note4}:
\begin{theorem} \label{theorem:model}
By excluding noninteracting Hamiltonians~\footnote{There are only two solutions whose Hamiltonians are noninteracting~\cite{Note4}.}, 
the solution of Eq.~\eqref{eq:Reduced_eigenstate-condition}
is restricted to the following three cases (and their equivalents obtained by appropriate permutations of the directions of the Pauli matrices):
\begin{enumerate}
    \item \label{model:period1}
        The EAP state $\ket{1;00}$ is an eigenstate of
        \begin{align}
            \hat{H}_{1}=\sum_{j=1}^{N}\bigl(
            &J^{xy}\hat{\sigma}_{j}^{x}\hat{\sigma}_{j+1}^{y}
            +J^{yx}\hat{\sigma}_{j}^{y}\hat{\sigma}_{j+1}^{x}
            \nonumber\\
            &+J^{yz}\hat{\sigma}_{j}^{y}\hat{\sigma}_{j+1}^{z}
            +J^{zy}\hat{\sigma}_{j}^{z}\hat{\sigma}_{j+1}^{y}
            +h^{y}\hat{\sigma}_{j}^{y}
            \bigr)
        \end{align}
        for arbitrary values of $J^{xy},J^{yx},J^{yz},J^{zy},h^{y}$.
    \item \label{model:period3}
        When $N/2$ is a multiple of three, the EAP state $\ket{3;01,10,11}$ (and its translations) is an eigenstate of
        \begin{align}
            \hat{H}_{2}=\sum_{j=1}^{N}\bigl(
            J^{xy}\hat{\sigma}_{j}^{x}\hat{\sigma}_{j+1}^{y}
            +J^{yz}\hat{\sigma}_{j}^{y}\hat{\sigma}_{j+1}^{z}
            \bigr)
        \end{align}
        for arbitrary values of $J^{xy},J^{yz}$.
    \item \label{model:period4}
        When $N/2$ is a multiple of four,
        the EAP state $\ket{4;00,01,10,11}$ (and its translations) is an eigenstate of 
        \begin{align}
            \hat{H}_{3}=\sum_{j=1}^{N}\bigl(
            J^{xx}\hat{\sigma}_{j}^{x}\hat{\sigma}_{j+1}^{x}
            +J^{yz}\hat{\sigma}_{j}^{y}\hat{\sigma}_{j+1}^{z}
            \bigr)
        \end{align}
        for arbitrary values of $J^{xx},J^{yz}$.
\end{enumerate}
\end{theorem}
The fact that there exist not so many solutions implies that our condition Eq.~\eqref{eq:eigenstate-condition} is strong enough for translation-invariant Hamiltonians.

Considering the importance of these models, we call the models described by $\hat{H}_{1}$, $\hat{H}_{2}$ and $\hat{H}_{3}$ Models~\ref{model:period1}, \ref{model:period3} and \ref{model:period4}, respectively. Note that, since Model~\ref{model:period3} is included in Model~\ref{model:period1}, the EAP state $\ket{1;00}$ is also an eigenstate of $\hat{H}_{2}$.

Furthermore, we can obtain the following theorem regarding the nonintegrability of Models~\ref{model:period3} and \ref{model:period4}~\cite{Note4}:
\begin{theorem} \label{theorem:nonintegrability}
For Model~\ref{model:period3} with $J^{xy},J^{yz} \neq 0$ and for Model~\ref{model:period4} with $J^{xx},J^{yz} \neq 0$, there exists no local conserved quantity other than a linear combination of the identity and the Hamiltonian (i.e., trivial one).
\end{theorem}
Because it is known that integrable systems have many [$\mathcal{O}(N)$ number of] nontrivial local conserved quantities~\cite{Baxter1982,Korepin1993,Takahashi1999}, the above theorem implies that these models are indeed nonintegrable.
In addition, since Model~\ref{model:period3} (which is $J^{yx}=J^{zy}=h^{y}=0$ case of Model~\ref{model:period1}) is nonintegrable for any nonzero $J^{xy},J^{yz}$, Model~\ref{model:period1} will be nonintegrable, at least for nonaccidental values of the parameters~\cite{Note4}.

Combining Theorems~\ref{theorem:model} and \ref{theorem:nonintegrability} with the fact that EAP states are thermal as explained above, we finally obtain
an analytic expression of a thermal energy eigenstate of a nonintegrable system.
This highly contrasts with the recent progress made by H.~Tasaki~\cite{Tasaki2024} in attempts to prove the ETH. He proved, only with respect to special observables, that all energy eigenstates of a certain noninteracting (integrable) system are indistinguishable from the thermal state. However, problems in nonintegrable systems remained elusive in his work. In addition, the restriction on observables is essential in his result because the integrable system has many local observables by which energy eigenstates can be distinguished from the thermal state. By contrast, our approach is to construct, in a nonintegrable system, an energy eigenstate that is indistinguishable from a thermal state with respect to any local observables.

\textit{Finite-temperature state---}
So far, we have only considered states at $\beta = 0$. Can we describe thermal states at $\beta \neq 0$ using the EAP state? Here, we provide a method for constructing finite-temperature states using the EAP state. It should be noted that Hamiltonians considered here are not restricted to those having the EAP state as an eigenstate, which we have considered thus far.

Consider a system with translationally invariant finite-range interactions. Suppose that $\hat{H}$ is a real matrix with respect to the products of eigenstates of $\hat{\sigma}_j^z$, $\{ \ket{00\cdots0}, \ket{10\cdots0}, \ket{01\cdots0}, \cdots \}$. This class includes many well studied models in statistical mechanics, such as the Ising model and the Heisenberg model.

Here, we utilize the EAP state characterized by $p_j=q_j=0$, i.e., $\ket{1;00}$ in the notation introduced in Eq.~\eqref{eq:alt-notation}. Then, as will be discussed below, the imaginary-time evolution of the EAP state
\begin{align}
    \ket{\beta} \propto e^{- \frac{1}{4} \beta \hat{H}} \ket{1;00}
\end{align}
is locally indistinguishable from the Gibbs state $\hat{\rho}^\mathrm{can}$ at the inverse temperature $\beta$ with an exponentially small error, although it is not an eigenstate of $\hat{H}$~\footnote{Some readers may wonder, if the EAP state $\ket{1;00}$ is an energy eigenstate, then it is invariant under the imaginary-time evolution, and hence $\ket{\beta}$ cannot describe a finite-temperature state. However, by using Theorem~\ref{theorem:eigenstate-condition}, it immediately follows that $\ket{1;00}$ cannot be an energy eigenstate provided that $\hat{H}$ is translation invariant (for single-site translations) and is a real matrix.}.

Let $\hat{O}$ be a local observable of interest. Since both $\ket{\beta}$ and $\hat{\rho}^\mathrm{can}$ are translation invariant, without loss of generality, we can assume that $\hat{O}$ has support around $j=N/4$. We then introduce an approximation $\ket{\tilde{\beta}}$ for $\ket{\beta}$ as
\begin{align}
    \ket{\tilde{\beta}} \propto e^{- \frac{1}{4} \beta [\hat{H}-\hat{H}_\mathrm{int}]} \ket{1;00}.
\end{align}
Here $\hat{H}_\mathrm{int}$ represents the interaction between the left half $\{1,2,\cdots,N/2\}$ and the right half $\{N/2+1,N/2+2,\cdots,N\}$ of the whole system, defined as a sum of interaction terms when $\hat{H}$ is decomposed into a linear combination of the Pauli strings as in Eq.~\eqref{eq:Hamiltonian}.
The imaginary-time evolution is expected to be stable against local perturbations. Indeed, the Gibbs state, which is proportional to the imaginary-time evolution operator, is not affected by perturbations on infinitely distant points because systems under consideration are one dimensional and do not exhibit the first-order phase transition at finite temperature~\cite{Note4}. Therefore, since $\hat{H}_\mathrm{int}$ is a sum of local observables defined around $j=1$ or $N/2$ at a distance of $\mathcal{O}(N)$ from the support of $\hat{O}$, it is expected that $\ket{\tilde{\beta}}$ approximate $\ket{\beta}$ around $j=N/4$ in the limit of $N\to\infty$, i.e.,
\begin{align}
    \lim_{N\to\infty} \braket{\tilde{\beta}|\hat{O}|\tilde{\beta}}
    = \lim_{N\to\infty} \braket{{\beta}|\hat{O}|{\beta}}.
    \label{eq:beta-tilde_beta}
\end{align}
Under this assumption, we have the following~\cite{Note4}:
\begin{theorem} \label{theorem:finite-temperature}
Suppose that the interactions are of finite range (which means the range is bounded by a constant independent of $N$) and translation invariant and that $\hat{H}$ is represented by a real matrix in the basis formed by products of $\ket{0}_j$ and $\ket{1}_j$. Let $\hat{O}$ be an arbitrary local observable. Then, for any $\beta<\infty$, if Eq.~\eqref{eq:beta-tilde_beta} is satisfied, it holds that
\begin{align}
    \lim_{N\to\infty} \braket{\beta|\hat{O}|\beta}
    = \lim_{N\to\infty} \mathrm{Tr} [ \hat{\rho}^\mathrm{can} \hat{O} ].
    \label{eq:beta-state_Gibbs}
\end{align}
\end{theorem}
In other words, $\ket{\beta}$ is a thermal pure state at finite temperature although it is not an energy eigenstate.
We should emphasize that this state is a pure state of a closed system, in contrast to a purified Gibbs state~\cite{Verstraete2004,Feiguin2005}, which is a pure state of an extended system involving an ancillary system. In the imaginary-time evolved EAP state $\ket{\beta}$, for any local subsystem, its complement plays the role of an ancilla while itself in the correct thermal equilibrium. This property is by no means trivial and relies on the conditions on the Hamiltonian in Theorem~\ref{theorem:finite-temperature}.

We finally confirm the validity of Eq.~\eqref{eq:beta-state_Gibbs} by numerically testing our prediction on the quantum Ising chain, defined by the Hamiltonian
\begin{align}
    \hat{H} = - \sum_{j=1}^{N} \hat{\sigma}_{j}^z \hat{\sigma}_{j+1}^z - h^x \sum_{j=1}^{N} \hat{\sigma}_{j}^x - h^z \sum_{j=1}^{N} \hat{\sigma}_{j}^z.
\end{align}
This model is known to be integrable when $h^z = 0$ and nonintegrable when $h^x \neq 0, h^z \neq 0$~\cite{Chiba2024}. To ensure that the results do not depend on integrability, we investigate both the integrable case ($h^x=1, h^z=0$) and nonintegrable case ($h^x=1, h^z=1$).
We plot in Fig.~\ref{fig:beta-state_Gibbs} the $N$ dependence of the difference in the expectation value of the transverse magnetization $\hat{m}^x = \frac{1}{N} \sum_{j=1}^{N} \hat{\sigma}_{j}^x$ between the Gibbs state and the imaginary-time evolved EAP state $\ket{\beta}$.
For the integrable case, the expectation value in the Gibbs state is evaluated at the thermodynamic limit using the exact solution.
For the nonintegrable case, it is evaluated for the same system size as that of $\ket{\beta}$ using the exact diagonalization.
In both cases,
$\ket{\beta}$
is obtained by applying the imaginary-time evolution operator, expanded using a Taylor series, to the EAP state.
We can confirm that the expectation value in $\ket{\beta}$ converges to the correct value. Furthermore, the convergence is exponentially fast. Thus, by utilizing the EAP state evolved in imaginary time, one can calculate the thermal equilibrium values of local observables.

Here, we discuss the relation with the cTPQ state~\cite{Sugiura2013}, which is also a type of thermal pure states. The cTPQ state is the imaginary-time evolution of a Haar random state. The Haar random state is (almost) maximally entangled and is locally indistinguishable from the maximally mixed state, as is the EAP state. However, while constructing the TPQ state requires an imaginary-time evolution for $\beta/2$, constructing $\ket{\beta}$ requires only half of that, $\beta/4$.
In addition, unlike the cTPQ state, $\ket{\beta}$ does not entail statistical uncertainties. This means that generating a single $\ket{\beta}$ suffices, without incurring any sampling costs.

\begin{figure}
    \centering
    \includegraphics[width=\linewidth]{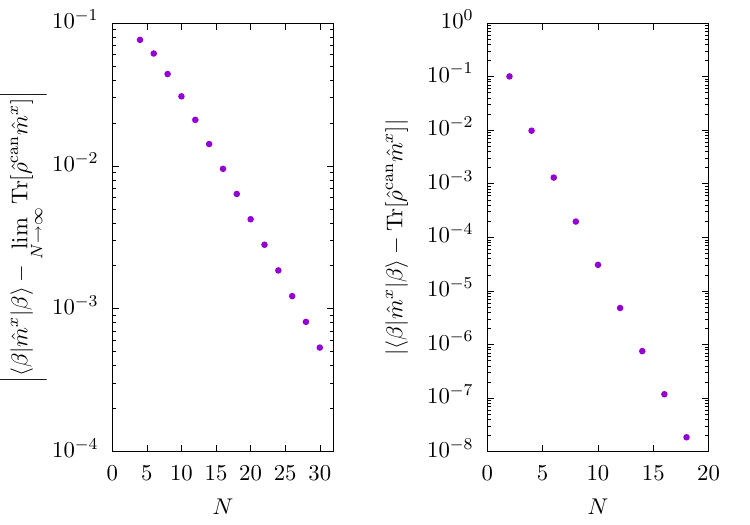}
    \caption{$N$ dependence of the difference in the expectation value of the transverse magnetization between the Gibbs state and the imaginary-time evolved EAP state for the quantum Ising model.
    (left) Integrable case ($h^x=1,h^z=0$), showing the difference from the exact value in the thermodynamic limit. (right) Nonintegrable case ($h^x=1,h^z=1$), showing the difference from the Gibbs state with the same system size.
    We set the inverse temperature to $\beta = 1$ in both cases.}
    \label{fig:beta-state_Gibbs}
\end{figure}

\textit{Discussion---}
We have studied a class of volume-law states, which we call entangled antipodal pair (EAP) states, and thoroughly characterized these states by providing the necessary and sufficient conditions that Hamiltonians must satisfy in order to have an EAP state as an eigenstate. Moreover, we have rigorously shown that some of such Hamiltonians are nonintegrable. Our EAP states are indistinguishable from the Gibbs state at infinite temperature. In other words, we have written down
analytic expressions for thermal energy eigenstates of nonintegrable many-body systems. We have also devised a method for constructing finite-temperature thermal states using an EAP state.

Some readers may be concerned that the energy eigenspace containing the EAP eigenstate is excessively large because, with an exponentially large number of states, it can be not so difficult to obtain a highly entangled state as their linear combination, even when each is hardly entangled~\cite{Wei2024}. However, additional numerical calculations~\cite{Note4} show that in Model~\ref{model:period1}, the degeneracy
at zero energy is only $2$ in the momentum sector~\footnote{Furthermore, by extending the locality of interactions from two to three, we have also obtained an example of a pair of the EAP eigenstate and the nonintegrable Hamiltonian that does not have degeneracy in the corresponding momentum sector~\cite{Note4}.}.
In addition, we can also confirm that all states in the degenerate eigenspace
have a volume-law entanglement with the maximal coefficient $\log 2$~\cite{Note4}. Hence, the nontriviality of our findings is not diminished by the degeneracy.

Our result suggests that theoretical analysis may be feasible even for thermal eigenstates of nonintegrable systems and may pave the way for giving a provable example of the ETH.
Since EAP states themselves can give only a few of eigenstates at infinite temperature, constructing general eigenstates remains challenging.
However, as we have shown, one can obtain thermal states at arbitrary temperature by applying an imaginary-time evolution, a low-complexity operation that can be well approximated via a matrix product operator with small bond dimension~\cite{Kuwahara2021}, to the EAP state. This implies that the expressivity of states derived from EAP states is quite high.
Thus, EAP states would be one of the most promising starting points for constructing general eigenstates, including those at finite temperature.
Furthermore, it will provide theoretical methodologies not only for thermalization, but also for various areas of quantum statistical mechanics that use thermal pure states, such as the formulation of statistical mechanics and finite temperature simulations.

\textit{Acknowledgments---} We are grateful to N.~Shiraishi, A.~Shimizu, H.~Katsura, K.~Fujii, K.~Mizuta, H.~Kono, M.~Yamaguchi, and Z.~Wei for useful discussions.
We would like to especially thank H.~K.~Park and A.~N.~Ivanov for careful reading and valuable comments on a draft of this Letter.
Y.Y. is supported by the Special Postdoctoral Researchers Program at RIKEN. Y.C. is supported by Japan Society for the Promotion of Science KAKENHI Grant No.~JP21J14313, the Special Postdoctoral Researchers Program at RIKEN, and JST ERATO Grant No.~JPMJER2302, Japan.

\bibliographystyle{apsrev4-2}
\bibliography{document}

\begin{thebibliography}{64}%
\makeatletter
\providecommand \@ifxundefined [1]{%
 \@ifx{#1\undefined}
}%
\providecommand \@ifnum [1]{%
 \ifnum #1\expandafter \@firstoftwo
 \else \expandafter \@secondoftwo
 \fi
}%
\providecommand \@ifx [1]{%
 \ifx #1\expandafter \@firstoftwo
 \else \expandafter \@secondoftwo
 \fi
}%
\providecommand \natexlab [1]{#1}%
\providecommand \enquote  [1]{``#1''}%
\providecommand \bibnamefont  [1]{#1}%
\providecommand \bibfnamefont [1]{#1}%
\providecommand \citenamefont [1]{#1}%
\providecommand \href@noop [0]{\@secondoftwo}%
\providecommand \href [0]{\begingroup \@sanitize@url \@href}%
\providecommand \@href[1]{\@@startlink{#1}\@@href}%
\providecommand \@@href[1]{\endgroup#1\@@endlink}%
\providecommand \@sanitize@url [0]{\catcode `\\12\catcode `\$12\catcode `\&12\catcode `\#12\catcode `\^12\catcode `\_12\catcode `\%12\relax}%
\providecommand \@@startlink[1]{}%
\providecommand \@@endlink[0]{}%
\providecommand \url  [0]{\begingroup\@sanitize@url \@url }%
\providecommand \@url [1]{\endgroup\@href {#1}{\urlprefix }}%
\providecommand \urlprefix  [0]{URL }%
\providecommand \Eprint [0]{\href }%
\providecommand \doibase [0]{https://doi.org/}%
\providecommand \selectlanguage [0]{\@gobble}%
\providecommand \bibinfo  [0]{\@secondoftwo}%
\providecommand \bibfield  [0]{\@secondoftwo}%
\providecommand \translation [1]{[#1]}%
\providecommand \BibitemOpen [0]{}%
\providecommand \bibitemStop [0]{}%
\providecommand \bibitemNoStop [0]{.\EOS\space}%
\providecommand \EOS [0]{\spacefactor3000\relax}%
\providecommand \BibitemShut  [1]{\csname bibitem#1\endcsname}%
\let\auto@bib@innerbib\@empty
%</preamble>
\bibitem [{\citenamefont {D'Alessio}\ \emph {et~al.}(2016)\citenamefont {D'Alessio}, \citenamefont {Kafri}, \citenamefont {Polkovnikov},\ and\ \citenamefont {Rigol}}]{DAlessio2016}%
  \BibitemOpen
  \bibfield  {author} {\bibinfo {author} {\bibfnamefont {L.}~\bibnamefont {D'Alessio}}, \bibinfo {author} {\bibfnamefont {Y.}~\bibnamefont {Kafri}}, \bibinfo {author} {\bibfnamefont {A.}~\bibnamefont {Polkovnikov}},\ and\ \bibinfo {author} {\bibfnamefont {M.}~\bibnamefont {Rigol}},\ }\href {https://doi.org/10.1080/00018732.2016.1198134} {\bibfield  {journal} {\bibinfo  {journal} {Adv. Phys.}\ }\textbf {\bibinfo {volume} {65}},\ \bibinfo {pages} {239} (\bibinfo {year} {2016})}\BibitemShut {NoStop}%
\bibitem [{\citenamefont {Gogolin}\ and\ \citenamefont {Eisert}(2016)}]{Gogolin2016}%
  \BibitemOpen
  \bibfield  {author} {\bibinfo {author} {\bibfnamefont {C.}~\bibnamefont {Gogolin}}\ and\ \bibinfo {author} {\bibfnamefont {J.}~\bibnamefont {Eisert}},\ }\href {https://doi.org/10.1088/0034-4885/79/5/056001} {\bibfield  {journal} {\bibinfo  {journal} {Rep. Prog. Phys.}\ }\textbf {\bibinfo {volume} {79}},\ \bibinfo {pages} {056001} (\bibinfo {year} {2016})}\BibitemShut {NoStop}%
\bibitem [{\citenamefont {Mori}\ \emph {et~al.}(2018)\citenamefont {Mori}, \citenamefont {Ikeda}, \citenamefont {Kaminishi},\ and\ \citenamefont {Ueda}}]{Mori2018}%
  \BibitemOpen
  \bibfield  {author} {\bibinfo {author} {\bibfnamefont {T.}~\bibnamefont {Mori}}, \bibinfo {author} {\bibfnamefont {T.~N.}\ \bibnamefont {Ikeda}}, \bibinfo {author} {\bibfnamefont {E.}~\bibnamefont {Kaminishi}},\ and\ \bibinfo {author} {\bibfnamefont {M.}~\bibnamefont {Ueda}},\ }\href {https://doi.org/10.1088/1361-6455/aabcdf} {\bibfield  {journal} {\bibinfo  {journal} {J. Phys. B}\ }\textbf {\bibinfo {volume} {51}},\ \bibinfo {pages} {112001} (\bibinfo {year} {2018})}\BibitemShut {NoStop}%
\bibitem [{\citenamefont {von Neumann}(1929)}]{Neumann1929}%
  \BibitemOpen
  \bibfield  {author} {\bibinfo {author} {\bibfnamefont {J.}~\bibnamefont {von Neumann}},\ }\href {https://doi.org/10.1007/BF01339852} {\bibfield  {journal} {\bibinfo  {journal} {Z. Phys.}\ }\textbf {\bibinfo {volume} {57}},\ \bibinfo {pages} {30} (\bibinfo {year} {1929})}\BibitemShut {NoStop}%
\bibitem [{\citenamefont {Deutsch}(1991)}]{Deutsch1991}%
  \BibitemOpen
  \bibfield  {author} {\bibinfo {author} {\bibfnamefont {J.~M.}\ \bibnamefont {Deutsch}},\ }\href {https://doi.org/10.1103/PhysRevA.43.2046} {\bibfield  {journal} {\bibinfo  {journal} {Phys. Rev. A}\ }\textbf {\bibinfo {volume} {43}},\ \bibinfo {pages} {2046} (\bibinfo {year} {1991})}\BibitemShut {NoStop}%
\bibitem [{\citenamefont {Srednicki}(1994)}]{Srednicki1994}%
  \BibitemOpen
  \bibfield  {author} {\bibinfo {author} {\bibfnamefont {M.}~\bibnamefont {Srednicki}},\ }\href {https://doi.org/10.1103/PhysRevE.50.888} {\bibfield  {journal} {\bibinfo  {journal} {Phys. Rev. E}\ }\textbf {\bibinfo {volume} {50}},\ \bibinfo {pages} {888} (\bibinfo {year} {1994})}\BibitemShut {NoStop}%
\bibitem [{\citenamefont {Rigol}\ \emph {et~al.}(2008)\citenamefont {Rigol}, \citenamefont {Dunjko},\ and\ \citenamefont {Olshanii}}]{Rigol2008}%
  \BibitemOpen
  \bibfield  {author} {\bibinfo {author} {\bibfnamefont {M.}~\bibnamefont {Rigol}}, \bibinfo {author} {\bibfnamefont {V.}~\bibnamefont {Dunjko}},\ and\ \bibinfo {author} {\bibfnamefont {M.}~\bibnamefont {Olshanii}},\ }\href {https://doi.org/10.1038/nature06838} {\bibfield  {journal} {\bibinfo  {journal} {Nature (London)}\ }\textbf {\bibinfo {volume} {452}},\ \bibinfo {pages} {854} (\bibinfo {year} {2008})}\BibitemShut {NoStop}%
\bibitem [{\citenamefont {Kim}\ \emph {et~al.}(2014)\citenamefont {Kim}, \citenamefont {Ikeda},\ and\ \citenamefont {Huse}}]{Kim2014}%
  \BibitemOpen
  \bibfield  {author} {\bibinfo {author} {\bibfnamefont {H.}~\bibnamefont {Kim}}, \bibinfo {author} {\bibfnamefont {T.~N.}\ \bibnamefont {Ikeda}},\ and\ \bibinfo {author} {\bibfnamefont {D.~A.}\ \bibnamefont {Huse}},\ }\href {https://doi.org/10.1103/PhysRevE.90.052105} {\bibfield  {journal} {\bibinfo  {journal} {Phys. Rev. E}\ }\textbf {\bibinfo {volume} {90}},\ \bibinfo {pages} {052105} (\bibinfo {year} {2014})}\BibitemShut {NoStop}%
\bibitem [{\citenamefont {Beugeling}\ \emph {et~al.}(2014)\citenamefont {Beugeling}, \citenamefont {Moessner},\ and\ \citenamefont {Haque}}]{Beugeling2014}%
  \BibitemOpen
  \bibfield  {author} {\bibinfo {author} {\bibfnamefont {W.}~\bibnamefont {Beugeling}}, \bibinfo {author} {\bibfnamefont {R.}~\bibnamefont {Moessner}},\ and\ \bibinfo {author} {\bibfnamefont {M.}~\bibnamefont {Haque}},\ }\href {https://doi.org/10.1103/PhysRevE.89.042112} {\bibfield  {journal} {\bibinfo  {journal} {Phys. Rev. E}\ }\textbf {\bibinfo {volume} {89}},\ \bibinfo {pages} {042112} (\bibinfo {year} {2014})}\BibitemShut {NoStop}%
\bibitem [{\citenamefont {Steinigeweg}\ \emph {et~al.}(2014)\citenamefont {Steinigeweg}, \citenamefont {Khodja}, \citenamefont {Niemeyer}, \citenamefont {Gogolin},\ and\ \citenamefont {Gemmer}}]{Steinigeweg2014}%
  \BibitemOpen
  \bibfield  {author} {\bibinfo {author} {\bibfnamefont {R.}~\bibnamefont {Steinigeweg}}, \bibinfo {author} {\bibfnamefont {A.}~\bibnamefont {Khodja}}, \bibinfo {author} {\bibfnamefont {H.}~\bibnamefont {Niemeyer}}, \bibinfo {author} {\bibfnamefont {C.}~\bibnamefont {Gogolin}},\ and\ \bibinfo {author} {\bibfnamefont {J.}~\bibnamefont {Gemmer}},\ }\href {https://doi.org/10.1103/PhysRevLett.112.130403} {\bibfield  {journal} {\bibinfo  {journal} {Phys. Rev. Lett.}\ }\textbf {\bibinfo {volume} {112}},\ \bibinfo {pages} {130403} (\bibinfo {year} {2014})}\BibitemShut {NoStop}%
\bibitem [{\citenamefont {Verstraete}\ \emph {et~al.}(2008)\citenamefont {Verstraete}, \citenamefont {Murg},\ and\ \citenamefont {Cirac}}]{Verstraete2008}%
  \BibitemOpen
  \bibfield  {author} {\bibinfo {author} {\bibfnamefont {F.}~\bibnamefont {Verstraete}}, \bibinfo {author} {\bibfnamefont {V.}~\bibnamefont {Murg}},\ and\ \bibinfo {author} {\bibfnamefont {J.}~\bibnamefont {Cirac}},\ }\href {https://doi.org/10.1080/14789940801912366} {\bibfield  {journal} {\bibinfo  {journal} {Adv. Phys.}\ }\textbf {\bibinfo {volume} {57}},\ \bibinfo {pages} {143} (\bibinfo {year} {2008})}\BibitemShut {NoStop}%
\bibitem [{\citenamefont {Schollw^^c3^^b6ck}(2011)}]{Schollwock2011}%
  \BibitemOpen
  \bibfield  {author} {\bibinfo {author} {\bibfnamefont {U.}~\bibnamefont {Schollw^^c3^^b6ck}},\ }\href {https://doi.org/10.1016/j.aop.2010.09.012} {\bibfield  {journal} {\bibinfo  {journal} {Ann. Phys.}\ }\textbf {\bibinfo {volume} {326}},\ \bibinfo {pages} {96} (\bibinfo {year} {2011})}\BibitemShut {NoStop}%
\bibitem [{\citenamefont {Moudgalya}\ \emph {et~al.}(2018{\natexlab{a}})\citenamefont {Moudgalya}, \citenamefont {Rachel}, \citenamefont {Bernevig},\ and\ \citenamefont {Regnault}}]{Moudgalya2018_1}%
  \BibitemOpen
  \bibfield  {author} {\bibinfo {author} {\bibfnamefont {S.}~\bibnamefont {Moudgalya}}, \bibinfo {author} {\bibfnamefont {S.}~\bibnamefont {Rachel}}, \bibinfo {author} {\bibfnamefont {B.~A.}\ \bibnamefont {Bernevig}},\ and\ \bibinfo {author} {\bibfnamefont {N.}~\bibnamefont {Regnault}},\ }\href {https://doi.org/10.1103/PhysRevB.98.235155} {\bibfield  {journal} {\bibinfo  {journal} {Phys. Rev. B}\ }\textbf {\bibinfo {volume} {98}},\ \bibinfo {pages} {235155} (\bibinfo {year} {2018}{\natexlab{a}})}\BibitemShut {NoStop}%
\bibitem [{\citenamefont {Moudgalya}\ \emph {et~al.}(2018{\natexlab{b}})\citenamefont {Moudgalya}, \citenamefont {Regnault},\ and\ \citenamefont {Bernevig}}]{Moudgalya2018_2}%
  \BibitemOpen
  \bibfield  {author} {\bibinfo {author} {\bibfnamefont {S.}~\bibnamefont {Moudgalya}}, \bibinfo {author} {\bibfnamefont {N.}~\bibnamefont {Regnault}},\ and\ \bibinfo {author} {\bibfnamefont {B.~A.}\ \bibnamefont {Bernevig}},\ }\href {https://doi.org/10.1103/PhysRevB.98.235156} {\bibfield  {journal} {\bibinfo  {journal} {Phys. Rev. B}\ }\textbf {\bibinfo {volume} {98}},\ \bibinfo {pages} {235156} (\bibinfo {year} {2018}{\natexlab{b}})}\BibitemShut {NoStop}%
\bibitem [{\citenamefont {Moudgalya}\ \emph {et~al.}(2020)\citenamefont {Moudgalya}, \citenamefont {O'Brien}, \citenamefont {Bernevig}, \citenamefont {Fendley},\ and\ \citenamefont {Regnault}}]{Moudgalya2020}%
  \BibitemOpen
  \bibfield  {author} {\bibinfo {author} {\bibfnamefont {S.}~\bibnamefont {Moudgalya}}, \bibinfo {author} {\bibfnamefont {E.}~\bibnamefont {O'Brien}}, \bibinfo {author} {\bibfnamefont {B.~A.}\ \bibnamefont {Bernevig}}, \bibinfo {author} {\bibfnamefont {P.}~\bibnamefont {Fendley}},\ and\ \bibinfo {author} {\bibfnamefont {N.}~\bibnamefont {Regnault}},\ }\href {https://doi.org/10.1103/PhysRevB.102.085120} {\bibfield  {journal} {\bibinfo  {journal} {Phys. Rev. B}\ }\textbf {\bibinfo {volume} {102}},\ \bibinfo {pages} {085120} (\bibinfo {year} {2020})}\BibitemShut {NoStop}%
\bibitem [{\citenamefont {Shiraishi}\ and\ \citenamefont {Mori}(2017)}]{Shiraishi2017}%
  \BibitemOpen
  \bibfield  {author} {\bibinfo {author} {\bibfnamefont {N.}~\bibnamefont {Shiraishi}}\ and\ \bibinfo {author} {\bibfnamefont {T.}~\bibnamefont {Mori}},\ }\href {https://doi.org/10.1103/PhysRevLett.119.030601} {\bibfield  {journal} {\bibinfo  {journal} {Phys. Rev. Lett.}\ }\textbf {\bibinfo {volume} {119}},\ \bibinfo {pages} {030601} (\bibinfo {year} {2017})}\BibitemShut {NoStop}%
\bibitem [{\citenamefont {Mori}\ and\ \citenamefont {Shiraishi}(2017)}]{Mori2017}%
  \BibitemOpen
  \bibfield  {author} {\bibinfo {author} {\bibfnamefont {T.}~\bibnamefont {Mori}}\ and\ \bibinfo {author} {\bibfnamefont {N.}~\bibnamefont {Shiraishi}},\ }\href {https://doi.org/10.1103/PhysRevE.96.022153} {\bibfield  {journal} {\bibinfo  {journal} {Phys. Rev. E}\ }\textbf {\bibinfo {volume} {96}},\ \bibinfo {pages} {022153} (\bibinfo {year} {2017})}\BibitemShut {NoStop}%
\bibitem [{\citenamefont {Bernien}\ \emph {et~al.}(2017)\citenamefont {Bernien}, \citenamefont {Schwartz}, \citenamefont {Keesling}, \citenamefont {Levine}, \citenamefont {Omran}, \citenamefont {Pichler}, \citenamefont {Choi}, \citenamefont {Zibrov}, \citenamefont {Endres}, \citenamefont {Greiner}, \citenamefont {Vuletic},\ and\ \citenamefont {Lukin}}]{Bernien2017}%
  \BibitemOpen
  \bibfield  {author} {\bibinfo {author} {\bibfnamefont {H.}~\bibnamefont {Bernien}}, \bibinfo {author} {\bibfnamefont {S.}~\bibnamefont {Schwartz}}, \bibinfo {author} {\bibfnamefont {A.}~\bibnamefont {Keesling}}, \bibinfo {author} {\bibfnamefont {H.}~\bibnamefont {Levine}}, \bibinfo {author} {\bibfnamefont {A.}~\bibnamefont {Omran}}, \bibinfo {author} {\bibfnamefont {H.}~\bibnamefont {Pichler}}, \bibinfo {author} {\bibfnamefont {S.}~\bibnamefont {Choi}}, \bibinfo {author} {\bibfnamefont {A.~S.}\ \bibnamefont {Zibrov}}, \bibinfo {author} {\bibfnamefont {M.}~\bibnamefont {Endres}}, \bibinfo {author} {\bibfnamefont {M.}~\bibnamefont {Greiner}}, \bibinfo {author} {\bibfnamefont {V.}~\bibnamefont {Vuletic}},\ and\ \bibinfo {author} {\bibfnamefont {M.~D.}\ \bibnamefont {Lukin}},\ }\href {https://doi.org/10.1038/nature24622} {\bibfield  {journal} {\bibinfo  {journal} {Nature (London)}\ }\textbf {\bibinfo {volume} {551}},\ \bibinfo {pages} {579} (\bibinfo {year} {2017})}\BibitemShut {NoStop}%
\bibitem [{\citenamefont {Turner}\ \emph {et~al.}(2018{\natexlab{a}})\citenamefont {Turner}, \citenamefont {Michailidis}, \citenamefont {Abanin}, \citenamefont {Serbyn},\ and\ \citenamefont {Papi{\'{c}}}}]{Turner2018a}%
  \BibitemOpen
  \bibfield  {author} {\bibinfo {author} {\bibfnamefont {C.~J.}\ \bibnamefont {Turner}}, \bibinfo {author} {\bibfnamefont {A.~A.}\ \bibnamefont {Michailidis}}, \bibinfo {author} {\bibfnamefont {D.~A.}\ \bibnamefont {Abanin}}, \bibinfo {author} {\bibfnamefont {M.}~\bibnamefont {Serbyn}},\ and\ \bibinfo {author} {\bibfnamefont {Z.}~\bibnamefont {Papi{\'{c}}}},\ }\href {https://doi.org/10.1038/s41567-018-0137-5} {\bibfield  {journal} {\bibinfo  {journal} {Nat. Phys.}\ }\textbf {\bibinfo {volume} {14}},\ \bibinfo {pages} {745} (\bibinfo {year} {2018}{\natexlab{a}})}\BibitemShut {NoStop}%
\bibitem [{\citenamefont {Turner}\ \emph {et~al.}(2018{\natexlab{b}})\citenamefont {Turner}, \citenamefont {Michailidis}, \citenamefont {Abanin}, \citenamefont {Serbyn},\ and\ \citenamefont {Papi{\'{c}}}}]{Turner2018b}%
  \BibitemOpen
  \bibfield  {author} {\bibinfo {author} {\bibfnamefont {C.~J.}\ \bibnamefont {Turner}}, \bibinfo {author} {\bibfnamefont {A.~A.}\ \bibnamefont {Michailidis}}, \bibinfo {author} {\bibfnamefont {D.~A.}\ \bibnamefont {Abanin}}, \bibinfo {author} {\bibfnamefont {M.}~\bibnamefont {Serbyn}},\ and\ \bibinfo {author} {\bibfnamefont {Z.}~\bibnamefont {Papi{\'{c}}}},\ }\href {https://doi.org/10.1103/PhysRevB.98.155134} {\bibfield  {journal} {\bibinfo  {journal} {Phys. Rev. B}\ }\textbf {\bibinfo {volume} {98}},\ \bibinfo {pages} {155134} (\bibinfo {year} {2018}{\natexlab{b}})}\BibitemShut {NoStop}%
\bibitem [{\citenamefont {Vitagliano}\ \emph {et~al.}(2010)\citenamefont {Vitagliano}, \citenamefont {Riera},\ and\ \citenamefont {Latorre}}]{Vitagliano2010}%
  \BibitemOpen
  \bibfield  {author} {\bibinfo {author} {\bibfnamefont {G.}~\bibnamefont {Vitagliano}}, \bibinfo {author} {\bibfnamefont {A.}~\bibnamefont {Riera}},\ and\ \bibinfo {author} {\bibfnamefont {J.~I.}\ \bibnamefont {Latorre}},\ }\href {https://doi.org/10.1088/1367-2630/12/11/113049} {\bibfield  {journal} {\bibinfo  {journal} {New J. Phys.}\ }\textbf {\bibinfo {volume} {12}},\ \bibinfo {pages} {113049} (\bibinfo {year} {2010})}\BibitemShut {NoStop}%
\bibitem [{\citenamefont {Ram^^c3^^adrez}\ \emph {et~al.}(2014)\citenamefont {Ram^^c3^^adrez}, \citenamefont {Rodr^^c3^^adguez-Laguna},\ and\ \citenamefont {Sierra}}]{Ramirez2014}%
  \BibitemOpen
  \bibfield  {author} {\bibinfo {author} {\bibfnamefont {G.}~\bibnamefont {Ram^^c3^^adrez}}, \bibinfo {author} {\bibfnamefont {J.}~\bibnamefont {Rodr^^c3^^adguez-Laguna}},\ and\ \bibinfo {author} {\bibfnamefont {G.}~\bibnamefont {Sierra}},\ }\href {https://doi.org/10.1088/1742-5468/2014/10/P10004} {\bibfield  {journal} {\bibinfo  {journal} {J. Stat. Mech.}\ }\textbf {\bibinfo {volume} {2014}},\ \bibinfo {pages} {P10004} (\bibinfo {year} {2014})}\BibitemShut {NoStop}%
\bibitem [{\citenamefont {Ram^^c3^^adrez}\ \emph {et~al.}(2015)\citenamefont {Ram^^c3^^adrez}, \citenamefont {Rodr^^c3^^adguez-Laguna},\ and\ \citenamefont {Sierra}}]{Ramirez2015}%
  \BibitemOpen
  \bibfield  {author} {\bibinfo {author} {\bibfnamefont {G.}~\bibnamefont {Ram^^c3^^adrez}}, \bibinfo {author} {\bibfnamefont {J.}~\bibnamefont {Rodr^^c3^^adguez-Laguna}},\ and\ \bibinfo {author} {\bibfnamefont {G.}~\bibnamefont {Sierra}},\ }\href {https://doi.org/10.1088/1742-5468/2015/06/P06002} {\bibfield  {journal} {\bibinfo  {journal} {J. Stat. Mech.}\ }\textbf {\bibinfo {volume} {2015}},\ \bibinfo {pages} {P06002} (\bibinfo {year} {2015})}\BibitemShut {NoStop}%
\bibitem [{\citenamefont {Langlett}\ \emph {et~al.}(2022)\citenamefont {Langlett}, \citenamefont {Yang}, \citenamefont {Wildeboer}, \citenamefont {Gorshkov}, \citenamefont {Iadecola},\ and\ \citenamefont {Xu}}]{Langlett2022}%
  \BibitemOpen
  \bibfield  {author} {\bibinfo {author} {\bibfnamefont {C.~M.}\ \bibnamefont {Langlett}}, \bibinfo {author} {\bibfnamefont {Z.-C.}\ \bibnamefont {Yang}}, \bibinfo {author} {\bibfnamefont {J.}~\bibnamefont {Wildeboer}}, \bibinfo {author} {\bibfnamefont {A.~V.}\ \bibnamefont {Gorshkov}}, \bibinfo {author} {\bibfnamefont {T.}~\bibnamefont {Iadecola}},\ and\ \bibinfo {author} {\bibfnamefont {S.}~\bibnamefont {Xu}},\ }\href {https://doi.org/10.1103/PhysRevB.105.L060301} {\bibfield  {journal} {\bibinfo  {journal} {Phys. Rev. B}\ }\textbf {\bibinfo {volume} {105}},\ \bibinfo {pages} {L060301} (\bibinfo {year} {2022})}\BibitemShut {NoStop}%
\bibitem [{\citenamefont {Bettaque}\ and\ \citenamefont {Swingle}(2024)}]{Bettaque2024}%
  \BibitemOpen
  \bibfield  {author} {\bibinfo {author} {\bibfnamefont {V.}~\bibnamefont {Bettaque}}\ and\ \bibinfo {author} {\bibfnamefont {B.}~\bibnamefont {Swingle}},\ }\href {https://doi.org/10.22331/q-2024-05-27-1362} {\bibfield  {journal} {\bibinfo  {journal} {Quantum}\ }\textbf {\bibinfo {volume} {8}},\ \bibinfo {pages} {1362} (\bibinfo {year} {2024})}\BibitemShut {NoStop}%
\bibitem [{Note1()}]{Note1}%
  \BibitemOpen
  \bibinfo {note} {Some special EAP states appear in the study of integrable systems~\cite {Caetano2022,Ekman2022}.}\BibitemShut {Stop}%
\bibitem [{Note2()}]{Note2}%
  \BibitemOpen
  \bibinfo {note} {It is possible to extend EAP states to higher-dimensional systems. However, unlike the one-dimensional case, there exists arbitrariness in the choice of antipodal sites.}\BibitemShut {Stop}%
\bibitem [{\citenamefont {Goldstein}\ \emph {et~al.}(2015)\citenamefont {Goldstein}, \citenamefont {Huse}, \citenamefont {Lebowitz},\ and\ \citenamefont {Tumulka}}]{Goldstein2015}%
  \BibitemOpen
  \bibfield  {author} {\bibinfo {author} {\bibfnamefont {S.}~\bibnamefont {Goldstein}}, \bibinfo {author} {\bibfnamefont {D.~A.}\ \bibnamefont {Huse}}, \bibinfo {author} {\bibfnamefont {J.~L.}\ \bibnamefont {Lebowitz}},\ and\ \bibinfo {author} {\bibfnamefont {R.}~\bibnamefont {Tumulka}},\ }\href {https://doi.org/10.1103/PhysRevLett.115.100402} {\bibfield  {journal} {\bibinfo  {journal} {Phys. Rev. Lett.}\ }\textbf {\bibinfo {volume} {115}},\ \bibinfo {pages} {100402} (\bibinfo {year} {2015})}\BibitemShut {NoStop}%
\bibitem [{\citenamefont {Goldstein}\ \emph {et~al.}(2017)\citenamefont {Goldstein}, \citenamefont {Huse}, \citenamefont {Lebowitz},\ and\ \citenamefont {Tumulka}}]{Goldstein2017}%
  \BibitemOpen
  \bibfield  {author} {\bibinfo {author} {\bibfnamefont {S.}~\bibnamefont {Goldstein}}, \bibinfo {author} {\bibfnamefont {D.~A.}\ \bibnamefont {Huse}}, \bibinfo {author} {\bibfnamefont {J.~L.}\ \bibnamefont {Lebowitz}},\ and\ \bibinfo {author} {\bibfnamefont {R.}~\bibnamefont {Tumulka}},\ }\href {https://doi.org/10.1002/andp.201600301} {\bibfield  {journal} {\bibinfo  {journal} {Ann. Phys. (Leipzig)}\ }\textbf {\bibinfo {volume} {529}},\ \bibinfo {pages} {1600301} (\bibinfo {year} {2017})}\BibitemShut {NoStop}%
\bibitem [{\citenamefont {Rigol}(2009{\natexlab{a}})}]{Rigol2009}%
  \BibitemOpen
  \bibfield  {author} {\bibinfo {author} {\bibfnamefont {M.}~\bibnamefont {Rigol}},\ }\href {https://doi.org/10.1103/PhysRevLett.103.100403} {\bibfield  {journal} {\bibinfo  {journal} {Phys. Rev. Lett.}\ }\textbf {\bibinfo {volume} {103}},\ \bibinfo {pages} {100403} (\bibinfo {year} {2009}{\natexlab{a}})}\BibitemShut {NoStop}%
\bibitem [{\citenamefont {Rigol}(2009{\natexlab{b}})}]{Rigol2009a}%
  \BibitemOpen
  \bibfield  {author} {\bibinfo {author} {\bibfnamefont {M.}~\bibnamefont {Rigol}},\ }\href {https://doi.org/10.1103/PhysRevA.80.053607} {\bibfield  {journal} {\bibinfo  {journal} {Phys. Rev. A}\ }\textbf {\bibinfo {volume} {80}},\ \bibinfo {pages} {053607} (\bibinfo {year} {2009}{\natexlab{b}})}\BibitemShut {NoStop}%
\bibitem [{\citenamefont {Santos}\ and\ \citenamefont {Rigol}(2010)}]{Santos2010}%
  \BibitemOpen
  \bibfield  {author} {\bibinfo {author} {\bibfnamefont {L.~F.}\ \bibnamefont {Santos}}\ and\ \bibinfo {author} {\bibfnamefont {M.}~\bibnamefont {Rigol}},\ }\href {https://doi.org/10.1103/PhysRevE.81.036206} {\bibfield  {journal} {\bibinfo  {journal} {Phys. Rev. E}\ }\textbf {\bibinfo {volume} {81}},\ \bibinfo {pages} {036206} (\bibinfo {year} {2010})}\BibitemShut {NoStop}%
\bibitem [{\citenamefont {Sugimoto}\ \emph {et~al.}(2023)\citenamefont {Sugimoto}, \citenamefont {Hamazaki},\ and\ \citenamefont {Ueda}}]{Sugimoto2023}%
  \BibitemOpen
  \bibfield  {author} {\bibinfo {author} {\bibfnamefont {S.}~\bibnamefont {Sugimoto}}, \bibinfo {author} {\bibfnamefont {R.}~\bibnamefont {Hamazaki}},\ and\ \bibinfo {author} {\bibfnamefont {M.}~\bibnamefont {Ueda}},\ }\href {https://arxiv.org/abs/2303.10069} {\bibinfo {title} {{Rigorous Bounds on Eigenstate Thermalization}}} (\bibinfo {year} {2023})\BibitemShut {NoStop}%
\bibitem [{Note3()}]{Note3}%
  \BibitemOpen
  \bibinfo {note} {In addition, we can show that EAP states are not thermal in a ``deep sense''~\cite {Ippoliti2023} (even when the moment $k$ of the projected ensemble is not so large, such as $k=2$) in contrast to typical eigenstates of nonintegrable systems~\cite {Cotler2023}.}\BibitemShut {Stop}%
\bibitem [{Note4()}]{Note4}%
  \BibitemOpen
  \bibinfo {note} {See Supplemental Material for proofs of theorems, additional numerical results, and discussions, which includes Refs.~\cite {Shiraishi2019,Park2024,Mehta2004,Atas2013,Araki1969,Araki1975}}\BibitemShut {NoStop}%
\bibitem [{\citenamefont {Udupa}\ \emph {et~al.}(2023)\citenamefont {Udupa}, \citenamefont {Sur}, \citenamefont {Nandy}, \citenamefont {Sen},\ and\ \citenamefont {Sen}}]{Udupa2023}%
  \BibitemOpen
  \bibfield  {author} {\bibinfo {author} {\bibfnamefont {A.}~\bibnamefont {Udupa}}, \bibinfo {author} {\bibfnamefont {S.}~\bibnamefont {Sur}}, \bibinfo {author} {\bibfnamefont {S.}~\bibnamefont {Nandy}}, \bibinfo {author} {\bibfnamefont {A.}~\bibnamefont {Sen}},\ and\ \bibinfo {author} {\bibfnamefont {D.}~\bibnamefont {Sen}},\ }\href {https://doi.org/10.1103/PhysRevB.108.214430} {\bibfield  {journal} {\bibinfo  {journal} {Phys. Rev. B}\ }\textbf {\bibinfo {volume} {108}},\ \bibinfo {pages} {214430} (\bibinfo {year} {2023})}\BibitemShut {NoStop}%
\bibitem [{\citenamefont {Ivanov}\ and\ \citenamefont {Motrunich}(2024)}]{Ivanov2024}%
  \BibitemOpen
  \bibfield  {author} {\bibinfo {author} {\bibfnamefont {A.~N.}\ \bibnamefont {Ivanov}}\ and\ \bibinfo {author} {\bibfnamefont {O.~I.}\ \bibnamefont {Motrunich}},\ }\href {http://arxiv.org/abs/2403.05515} {\bibinfo {title} {{Volume-entangled exact eigenstates in the PXP and related models in any dimension}}} (\bibinfo {year} {2024})\BibitemShut {NoStop}%
\bibitem [{Note5()}]{Note5}%
  \BibitemOpen
  \bibinfo {note} {Note that, for some of states~(\ref {eq:alt-notation}), $n$-site shift can cause a phase change, such as $\protect \hat {\protect \mathcal {T}}^3\mathinner {|{3;01,10,11}\rangle }=-\mathinner {|{3;01,10,11}\rangle }$, where $\protect \hat {\protect \mathcal {T}}$ is the one-site translation operator.}\BibitemShut {Stop}%
\bibitem [{Note6()}]{Note6}%
  \BibitemOpen
  \bibinfo {note} {There are only two solutions whose Hamiltonians are noninteracting~\cite {Note4}.}\BibitemShut {Stop}%
\bibitem [{\citenamefont {Baxter}(1982)}]{Baxter1982}%
  \BibitemOpen
  \bibfield  {author} {\bibinfo {author} {\bibfnamefont {R.~J.}\ \bibnamefont {Baxter}},\ }\href@noop {} {\emph {\bibinfo {title} {{Exactly Solved Models in Statistical Mechanics}}}}\ (\bibinfo  {publisher} {Academic Press},\ \bibinfo {address} {San Diego},\ \bibinfo {year} {1982})\BibitemShut {NoStop}%
\bibitem [{\citenamefont {Korepin}\ \emph {et~al.}(1993)\citenamefont {Korepin}, \citenamefont {Bogoliubov},\ and\ \citenamefont {Izergin}}]{Korepin1993}%
  \BibitemOpen
  \bibfield  {author} {\bibinfo {author} {\bibfnamefont {V.~E.}\ \bibnamefont {Korepin}}, \bibinfo {author} {\bibfnamefont {N.~M.}\ \bibnamefont {Bogoliubov}},\ and\ \bibinfo {author} {\bibfnamefont {A.~G.}\ \bibnamefont {Izergin}},\ }\href {https://doi.org/10.1017/CBO9780511628832} {\emph {\bibinfo {title} {{Quantum Inverse Scattering Method and Correlation Functions}}}}\ (\bibinfo  {publisher} {Cambridge University Press},\ \bibinfo {address} {Cambridge},\ \bibinfo {year} {1993})\BibitemShut {NoStop}%
\bibitem [{\citenamefont {Takahashi}(1999)}]{Takahashi1999}%
  \BibitemOpen
  \bibfield  {author} {\bibinfo {author} {\bibfnamefont {M.}~\bibnamefont {Takahashi}},\ }\href {https://doi.org/10.1017/CBO9780511524332} {\emph {\bibinfo {title} {{Thermodynamics of One-Dimensional Solvable Models}}}}\ (\bibinfo  {publisher} {Cambridge University Press},\ \bibinfo {address} {Cambridge},\ \bibinfo {year} {1999})\BibitemShut {NoStop}%
\bibitem [{\citenamefont {Tasaki}(2024)}]{Tasaki2024}%
  \BibitemOpen
  \bibfield  {author} {\bibinfo {author} {\bibfnamefont {H.}~\bibnamefont {Tasaki}},\ }\href {https://doi.org/10.48550/arXiv.2401.15263} {\bibinfo {title} {{Macroscopic Irreversibility in Quantum Systems: ETH and Equilibration in a Free Fermion Chain}}} (\bibinfo {year} {2024})\BibitemShut {NoStop}%
\bibitem [{Note7()}]{Note7}%
  \BibitemOpen
  \bibinfo {note} {Some readers may wonder, if the EAP state $\mathinner {|{1;00}\rangle }$ is an energy eigenstate, then it is invariant under the imaginary-time evolution, and hence $\mathinner {|{\beta }\rangle }$ cannot describe a finite-temperature state. However, by using Theorem~\ref {theorem:eigenstate-condition}, it immediately follows that $\mathinner {|{1;00}\rangle }$ cannot be an energy eigenstate provided that $\protect \hat {H}$ is translation invariant (for single-site translations) and is a real matrix.}\BibitemShut {Stop}%
\bibitem [{\citenamefont {Verstraete}\ \emph {et~al.}(2004)\citenamefont {Verstraete}, \citenamefont {Garc\'\i{}a-Ripoll},\ and\ \citenamefont {Cirac}}]{Verstraete2004}%
  \BibitemOpen
  \bibfield  {author} {\bibinfo {author} {\bibfnamefont {F.}~\bibnamefont {Verstraete}}, \bibinfo {author} {\bibfnamefont {J.~J.}\ \bibnamefont {Garc\'\i{}a-Ripoll}},\ and\ \bibinfo {author} {\bibfnamefont {J.~I.}\ \bibnamefont {Cirac}},\ }\href {https://doi.org/10.1103/PhysRevLett.93.207204} {\bibfield  {journal} {\bibinfo  {journal} {Phys. Rev. Lett.}\ }\textbf {\bibinfo {volume} {93}},\ \bibinfo {pages} {207204} (\bibinfo {year} {2004})}\BibitemShut {NoStop}%
\bibitem [{\citenamefont {Feiguin}\ and\ \citenamefont {White}(2005)}]{Feiguin2005}%
  \BibitemOpen
  \bibfield  {author} {\bibinfo {author} {\bibfnamefont {A.~E.}\ \bibnamefont {Feiguin}}\ and\ \bibinfo {author} {\bibfnamefont {S.~R.}\ \bibnamefont {White}},\ }\href {https://doi.org/10.1103/PhysRevB.72.220401} {\bibfield  {journal} {\bibinfo  {journal} {Phys. Rev. B}\ }\textbf {\bibinfo {volume} {72}},\ \bibinfo {pages} {220401(R)} (\bibinfo {year} {2005})}\BibitemShut {NoStop}%
\bibitem [{\citenamefont {Chiba}(2024)}]{Chiba2024}%
  \BibitemOpen
  \bibfield  {author} {\bibinfo {author} {\bibfnamefont {Y.}~\bibnamefont {Chiba}},\ }\href {https://doi.org/10.1103/PhysRevB.109.035123} {\bibfield  {journal} {\bibinfo  {journal} {Phys. Rev. B}\ }\textbf {\bibinfo {volume} {109}},\ \bibinfo {pages} {035123} (\bibinfo {year} {2024})}\BibitemShut {NoStop}%
\bibitem [{\citenamefont {Sugiura}\ and\ \citenamefont {Shimizu}(2013)}]{Sugiura2013}%
  \BibitemOpen
  \bibfield  {author} {\bibinfo {author} {\bibfnamefont {S.}~\bibnamefont {Sugiura}}\ and\ \bibinfo {author} {\bibfnamefont {A.}~\bibnamefont {Shimizu}},\ }\href {https://doi.org/10.1103/PhysRevLett.111.010401} {\bibfield  {journal} {\bibinfo  {journal} {Phys. Rev. Lett.}\ }\textbf {\bibinfo {volume} {111}},\ \bibinfo {pages} {010401} (\bibinfo {year} {2013})}\BibitemShut {NoStop}%
\bibitem [{\citenamefont {Wei}\ and\ \citenamefont {Yoneta}(2024)}]{Wei2024}%
  \BibitemOpen
  \bibfield  {author} {\bibinfo {author} {\bibfnamefont {Z.}~\bibnamefont {Wei}}\ and\ \bibinfo {author} {\bibfnamefont {Y.}~\bibnamefont {Yoneta}},\ }\href {https://doi.org/10.1007/JHEP05(2024)251} {\bibfield  {journal} {\bibinfo  {journal} {J. High Energy Phys.}\ }\textbf {\bibinfo {volume} {2024}}\bibinfo  {number} { (05)},\ \bibinfo {pages} {251}}\BibitemShut {NoStop}%
\bibitem [{Note8()}]{Note8}%
  \BibitemOpen
\bibfield  {number} {  }\bibinfo {note} {Furthermore, by extending the locality of interactions from two to three, we have also obtained an example of a pair of the EAP eigenstate and the nonintegrable Hamiltonian that does not have degeneracy in the corresponding momentum sector~\cite {Note4}.}\BibitemShut {Stop}%
\bibitem [{\citenamefont {Kuwahara}\ \emph {et~al.}(2021)\citenamefont {Kuwahara}, \citenamefont {Alhambra},\ and\ \citenamefont {Anshu}}]{Kuwahara2021}%
  \BibitemOpen
  \bibfield  {author} {\bibinfo {author} {\bibfnamefont {T.}~\bibnamefont {Kuwahara}}, \bibinfo {author} {\bibfnamefont {A.~M.}\ \bibnamefont {Alhambra}},\ and\ \bibinfo {author} {\bibfnamefont {A.}~\bibnamefont {Anshu}},\ }\href {https://doi.org/10.1103/PhysRevX.11.011047} {\bibfield  {journal} {\bibinfo  {journal} {Phys. Rev. X}\ }\textbf {\bibinfo {volume} {11}},\ \bibinfo {pages} {011047} (\bibinfo {year} {2021})}\BibitemShut {NoStop}%
\bibitem [{\citenamefont {Caetano}\ and\ \citenamefont {Komatsu}(2022)}]{Caetano2022}%
  \BibitemOpen
  \bibfield  {author} {\bibinfo {author} {\bibfnamefont {J.}~\bibnamefont {Caetano}}\ and\ \bibinfo {author} {\bibfnamefont {S.}~\bibnamefont {Komatsu}},\ }\href {https://doi.org/10.1007/s10955-022-02914-6} {\bibfield  {journal} {\bibinfo  {journal} {J. Stat. Phys.}\ }\textbf {\bibinfo {volume} {187}},\ \bibinfo {pages} {30} (\bibinfo {year} {2022})}\BibitemShut {NoStop}%
\bibitem [{\citenamefont {Ekman}(2022)}]{Ekman2022}%
  \BibitemOpen
  \bibfield  {author} {\bibinfo {author} {\bibfnamefont {C.}~\bibnamefont {Ekman}},\ }\href {http://arxiv.org/abs/2207.12354} {\bibinfo {title} {{Crosscap states in the XXX spin-1/2 spin chain}}} (\bibinfo {year} {2022})\BibitemShut {NoStop}%
\bibitem [{\citenamefont {Ippoliti}\ and\ \citenamefont {Ho}(2023)}]{Ippoliti2023}%
  \BibitemOpen
  \bibfield  {author} {\bibinfo {author} {\bibfnamefont {M.}~\bibnamefont {Ippoliti}}\ and\ \bibinfo {author} {\bibfnamefont {W.~W.}\ \bibnamefont {Ho}},\ }\href {https://doi.org/10.1103/PRXQuantum.4.030322} {\bibfield  {journal} {\bibinfo  {journal} {PRX Quantum}\ }\textbf {\bibinfo {volume} {4}},\ \bibinfo {pages} {030322} (\bibinfo {year} {2023})}\BibitemShut {NoStop}%
\bibitem [{\citenamefont {Cotler}\ \emph {et~al.}(2023)\citenamefont {Cotler}, \citenamefont {Mark}, \citenamefont {Huang}, \citenamefont {Hern\'andez}, \citenamefont {Choi}, \citenamefont {Shaw}, \citenamefont {Endres},\ and\ \citenamefont {Choi}}]{Cotler2023}%
  \BibitemOpen
  \bibfield  {author} {\bibinfo {author} {\bibfnamefont {J.~S.}\ \bibnamefont {Cotler}}, \bibinfo {author} {\bibfnamefont {D.~K.}\ \bibnamefont {Mark}}, \bibinfo {author} {\bibfnamefont {H.-Y.}\ \bibnamefont {Huang}}, \bibinfo {author} {\bibfnamefont {F.}~\bibnamefont {Hern\'andez}}, \bibinfo {author} {\bibfnamefont {J.}~\bibnamefont {Choi}}, \bibinfo {author} {\bibfnamefont {A.~L.}\ \bibnamefont {Shaw}}, \bibinfo {author} {\bibfnamefont {M.}~\bibnamefont {Endres}},\ and\ \bibinfo {author} {\bibfnamefont {S.}~\bibnamefont {Choi}},\ }\href {https://doi.org/10.1103/PRXQuantum.4.010311} {\bibfield  {journal} {\bibinfo  {journal} {PRX Quantum}\ }\textbf {\bibinfo {volume} {4}},\ \bibinfo {pages} {010311} (\bibinfo {year} {2023})}\BibitemShut {NoStop}%
\bibitem [{\citenamefont {Shiraishi}(2019)}]{Shiraishi2019}%
  \BibitemOpen
  \bibfield  {author} {\bibinfo {author} {\bibfnamefont {N.}~\bibnamefont {Shiraishi}},\ }\href {https://doi.org/10.1209/0295-5075/128/17002} {\bibfield  {journal} {\bibinfo  {journal} {Europhys. Lett.}\ }\textbf {\bibinfo {volume} {128}},\ \bibinfo {pages} {17002} (\bibinfo {year} {2019})}\BibitemShut {NoStop}%
\bibitem [{\citenamefont {Park}\ and\ \citenamefont {Lee}(2024)}]{Park2024}%
  \BibitemOpen
  \bibfield  {author} {\bibinfo {author} {\bibfnamefont {H.~K.}\ \bibnamefont {Park}}\ and\ \bibinfo {author} {\bibfnamefont {S.}~\bibnamefont {Lee}},\ }\href {http://arxiv.org/abs/2403.02335} {\bibinfo {title} {{Proof of the nonintegrability of PXP model and general spin-$1/2$ systems}}} (\bibinfo {year} {2024})\BibitemShut {NoStop}%
\bibitem [{\citenamefont {Mehta}(2004)}]{Mehta2004}%
  \BibitemOpen
  \bibfield  {author} {\bibinfo {author} {\bibfnamefont {M.~L.}\ \bibnamefont {Mehta}},\ }\href@noop {} {\emph {\bibinfo {title} {{Random Matrices}}}},\ \bibinfo {edition} {3rd}\ ed.\ (\bibinfo  {publisher} {Academic Press},\ \bibinfo {address} {New York},\ \bibinfo {year} {2004})\BibitemShut {NoStop}%
\bibitem [{\citenamefont {Atas}\ \emph {et~al.}(2013)\citenamefont {Atas}, \citenamefont {Bogomolny}, \citenamefont {Giraud},\ and\ \citenamefont {Roux}}]{Atas2013}%
  \BibitemOpen
  \bibfield  {author} {\bibinfo {author} {\bibfnamefont {Y.~Y.}\ \bibnamefont {Atas}}, \bibinfo {author} {\bibfnamefont {E.}~\bibnamefont {Bogomolny}}, \bibinfo {author} {\bibfnamefont {O.}~\bibnamefont {Giraud}},\ and\ \bibinfo {author} {\bibfnamefont {G.}~\bibnamefont {Roux}},\ }\href {https://doi.org/10.1103/PhysRevLett.110.084101} {\bibfield  {journal} {\bibinfo  {journal} {Phys. Rev. Lett.}\ }\textbf {\bibinfo {volume} {110}},\ \bibinfo {pages} {084101} (\bibinfo {year} {2013})}\BibitemShut {NoStop}%
\bibitem [{\citenamefont {Araki}(1969)}]{Araki1969}%
  \BibitemOpen
  \bibfield  {author} {\bibinfo {author} {\bibfnamefont {H.}~\bibnamefont {Araki}},\ }\href {https://doi.org/10.1007/BF01645134} {\bibfield  {journal} {\bibinfo  {journal} {Commun. Math. Phys.}\ }\textbf {\bibinfo {volume} {14}},\ \bibinfo {pages} {120} (\bibinfo {year} {1969})}\BibitemShut {NoStop}%
\bibitem [{\citenamefont {Araki}(1975)}]{Araki1975}%
  \BibitemOpen
  \bibfield  {author} {\bibinfo {author} {\bibfnamefont {H.}~\bibnamefont {Araki}},\ }\href {https://doi.org/10.1007/BF01609054} {\bibfield  {journal} {\bibinfo  {journal} {Commun. Math. Phys.}\ }\textbf {\bibinfo {volume} {44}},\ \bibinfo {pages} {1} (\bibinfo {year} {1975})}\BibitemShut {NoStop}%
\bibitem [{Note9()}]{Note9}%
  \BibitemOpen
  \bibinfo {note} {To remove the possibility of accidental discrete symmetries, we avoid eigenspace of special momentum such as $0$ and $\pi $.}\BibitemShut {Stop}%
\bibitem [{Note10()}]{Note10}%
  \BibitemOpen
  \bibinfo {note} {We can always take $J^{zy}=0$ by an appropriate rotation around $y$ axis.}\BibitemShut {Stop}%
\bibitem [{Note11()}]{Note11}%
  \BibitemOpen
  \bibinfo {note} {The results of exact diagonalization also show that, at least for $N=8,16,20$, degeneracy at $E=0$ in the whole Hilbert space remains only two. This means that all eigenstates of the Hamiltonian~(\ref {eq:H_period2_11_NNN}) with eigenenergy $E=0$ are given by (linear combinations of) the EAP states $\mathinner {|{2;11,10}\rangle }, \mathinner {|{2;10,11}\rangle }$.}\BibitemShut {Stop}%
\end{thebibliography}%

\newcommand{\beginsupplement}{%
	\setcounter{table}{0}
	\renewcommand{\thetable}{S\arabic{table}}%
	\setcounter{figure}{0}
	\renewcommand{\thefigure}{S\arabic{figure}}%
	\setcounter{section}{0}
	\renewcommand{\thesection}{\Roman{section}}%
	\setcounter{equation}{0}
	\renewcommand{\theequation}{S\arabic{equation}}%
    \setcounter{page}{1}
}
\clearpage
\onecolumngrid
\beginsupplement

\begin{center}
	\textbf{\large Supplemental Material for\\ ``Exact Thermal Eigenstates of Nonintegrable Spin Chains at Infinite Temperature''}
\end{center}
\vspace{2mm}

\section{Proof of Theorem~\ref{theorem:eigenstate-condition}}
For the proof of Theorem~\ref{theorem:eigenstate-condition}, the following lemma is crucial:
\begin{lemma} \label{lemma:ExpValue_X_Y_ell}
Let $X$ and $Y$ be subsets of $\Lambda$ satisfying $D(X)< N/4$ and $D(Y)< N/4$. For any EAP state $\ket{\mathrm{EAP}}$, we have
\begin{align}
    \braket{\mathrm{EAP}|
        \bigotimes_{j \in X} \hat{\sigma}_{j}^{\mu_j}
        \bigotimes_{k \in Y} \hat{\sigma}_{k}^{\nu_k}
    |\mathrm{EAP}}
    = \begin{cases}
        1 & \text{when $Y=X$ and $\nu_{k}=\mu_{k}$ for all $k \in Y$}\\
        \displaystyle \prod_{j \in X} \omega_{j}^{\mu_{j}} & \text{when $Y=X+N/2$ and $\nu_{k}=\mu_{k-N/2}$ for all $k \in Y$}\\
        0 & \text{otherwise}
    \end{cases}.
    \label{eq:ExpValue_X_Y_ell}
\end{align}
\end{lemma}
\begin{proof}
We divide the lattice $\Lambda$ into four equally sized parts, $\Lambda_{a}:=\{aN/4+1,aN/4+2,...,(a+1)N/4\}$ $(a=0,1,2,3)$.
Since $D(X)<N/4$, without loss of generality, we can take $X\subset \Lambda_{0}$.
Since $D(Y)<N/4$, $Y$ cannot have any intersection with both $\Lambda_{0}$ and $\Lambda_{2}$,
or with both $\Lambda_{1}$ and $\Lambda_{3}$.

Suppose that $Y\cap\Lambda_{1}\neq \emptyset$ and let $k^*\in Y\cap\Lambda_{1}$. 
Because no Pauli operator acts on its antipodal site $k^*+N/2\in \Lambda_{3}$ in the left hand side of Eq.~(\ref{eq:ExpValue_X_Y_ell}), we have
\begin{align}
    &\braket{\mathrm{EAP}|
        \bigotimes_{j \in X} \hat{\sigma}_{j}^{\mu_j}
        \bigotimes_{k \in Y} \hat{\sigma}_{k}^{\nu_k}
    |\mathrm{EAP}}
    = \braket{\mathrm{EAP}|
        \left(
            \bigotimes_{j \in X} \hat{\sigma}_{j}^{\mu_j}
        \right)
        \left(
            \bigotimes_{k \in Y\setminus \{k^*\}} \hat{\sigma}_{k}^{\nu_k}
        \right)
        \hat{\sigma}_{k^*}^{\nu_{k^*}}
    |\mathrm{EAP}}
    =0.
\end{align}
If $Y\cap\Lambda_{3}\neq \emptyset$, we can obtain the same result in almost the same manner.
Therefore, in the following, we only need to consider the two cases $Y\subset\Lambda_{0}$ and $Y\subset\Lambda_{2}$.

Next we consider the case of $Y\subset\Lambda_{0}$. 
In the left hand side of Eq.~(\ref{eq:ExpValue_X_Y_ell}), at most two Pauli operators act on each site $j\in\Lambda_{0}$, and no Pauli operator acts on its antipodal site $j+N/2\in \Lambda_{2}$.
Therefore, unless $Y=X$ and $\nu_{k}=\mu_{k}$ for all $k\in Y$, 
the left hand side of Eq.~(\ref{eq:ExpValue_X_Y_ell}) becomes zero.
On the other hand, if $Y=X$ and $\nu_{k}=\mu_{k}$ for all $k \in Y$, we obviously have
\begin{align}
    \braket{\mathrm{EAP}|
        \bigotimes_{j \in X} \hat{\sigma}_{j}^{\mu_j}
        \bigotimes_{k \in Y} \hat{\sigma}_{k}^{\nu_k}
    |\mathrm{EAP}}
    = \braket{\mathrm{EAP}|\mathrm{EAP}}
    = 1.
\end{align}

Finally we consider the case of $Y\subset\Lambda_{2}$. 
In the left hand side of Eq.~(\ref{eq:ExpValue_X_Y_ell}), at most one Pauli operator acts on each site $j\in\Lambda_{0}$, and at most one Pauli operator acts on its antipodal site $j+N/2\in \Lambda_{2}$.
Using Eq.~\eqref{eq:sigma-EAP} of the main text, we have
\begin{align}
    \braket{\mathrm{EAP}|
        \bigotimes_{j \in X} \hat{\sigma}_{j}^{\mu_j}
        \bigotimes_{k \in Y} \hat{\sigma}_{k}^{\nu_k}
    |\mathrm{EAP}}
    = \braket{\mathrm{EAP}|
        \bigotimes_{j \in X} \hat{\sigma}_{j}^{\mu_j}
        \bigotimes_{k \in Y} \hat{\sigma}_{k-N/2}^{\nu_k}
    |\mathrm{EAP}}
    \prod_{k \in Y} \omega_k^{\nu_k}.
\end{align}
Applying the arguments of the previous paragraph, we can obtain the following: Unless $Y-N/2=X$ and $\nu_{k}=\mu_{k-N/2}$ for all $k\in Y$, 
the left hand side of Eq.~(\ref{eq:ExpValue_X_Y_ell}) becomes zero.
On the other hand, if $Y-N/2=X$ and $\nu_{k}=\mu_{k-N/2}$ for all $k\in Y$, we have
\begin{align}
    \braket{\mathrm{EAP}|
        \bigotimes_{j \in X} \hat{\sigma}_{j}^{\mu_j}
        \bigotimes_{k \in Y} \hat{\sigma}_{k}^{\nu_k}
    |\mathrm{EAP}}
    = \braket{\mathrm{EAP}|\mathrm{EAP}}\prod_{k \in Y} \omega_k^{\nu_k}
    = \prod_{j \in X} \omega_j^{\mu_j}.
\end{align}
\end{proof}

\begin{proof}[Proof of Theorem~\ref{theorem:eigenstate-condition}]
In order to show that three statements are equivalent, we need to show that 
(\ref{statement:eigenstate-condition_1}) $\Rightarrow$ (\ref{statement:eigenstate-condition_2}), 
(\ref{statement:eigenstate-condition_2}) $\Rightarrow$ (\ref{statement:eigenstate-condition_3}), 
and (\ref{statement:eigenstate-condition_3}) $\Rightarrow$ (\ref{statement:eigenstate-condition_1}).

\textit{(\ref{statement:eigenstate-condition_1}) $\Rightarrow$ (\ref{statement:eigenstate-condition_2}):}
We assume that an EAP state $\ket{\mathrm{EAP}}$ is an eigenstate of $\hat{H}$ with eigenvalue $\lambda$, i.e.,
\begin{align}
    \hat{H} \ket{\mathrm{EAP}} = \lambda \ket{\mathrm{EAP}}.
\end{align}
Then, taking the inner product with $\ket{\mathrm{EAP}}$ and substituting Eq.~\eqref{eq:Hamiltonian}, we have
\begin{align}
    \lambda
    = \braket{\mathrm{EAP}|\hat{H}|\mathrm{EAP}}
    = \sum_{X(\subset\Lambda)} \sum_{\vec{\mu}\in\{x,y,z\}^{X}}
    J_{X}^{\vec{\mu}}
    \braket{\mathrm{EAP}|\bigotimes_{j \in X} \hat{\sigma}_{j}^{\mu_{j}}|\mathrm{EAP}}
    = 0,
\end{align}
where the last equality follows from Lemma~\ref{lemma:ExpValue_X_Y_ell}.

\textit{(\ref{statement:eigenstate-condition_2}) $\Rightarrow$ (\ref{statement:eigenstate-condition_3}):}
We assume that an EAP state $\ket{\mathrm{EAP}}$ is an eigenstate of $\hat{H}$ with eigenvalue $0$, i.e.,
\begin{align}
    \hat{H} \ket{\mathrm{EAP}} = 0.
\end{align}
Then, substituting Eq.~\eqref{eq:Hamiltonian}, we have
\begin{align}
    \sum_{Y(\subset\Lambda)} \sum_{\vec{\nu}\in\{x,y,z\}^{Y}}
    J_{Y}^{\vec{\nu}} \bigotimes_{k \in Y} \hat{\sigma}_{k}^{\nu_{k}} \ket{\mathrm{EAP}}
    = 0.
\end{align}
Taking the inner product with $\displaystyle \bra{\mathrm{EAP}} \bigotimes_{j \in X} \hat{\sigma}_{j}^{\mu_{j}}$, we get
\begin{align}
    J_{X}^{\vec{\mu}} + J_{Y}^{\vec{\nu}} \prod_{j \in X} \omega_{j}^{\mu_{j}} = 0,
\end{align}
where $Y=X+N/2$ and $\nu_{k}=\mu_{k-N/2}$ for $k \in Y$.
Since $(\omega_j^{\mu_j})^{-1} = \omega_j^{\mu_j}$,
this implies Eq.~\eqref{eq:eigenstate-condition}.

\textit{(\ref{statement:eigenstate-condition_3}) $\Rightarrow$ (\ref{statement:eigenstate-condition_1}):}
We assume that Eq.~\eqref{eq:eigenstate-condition} holds.
Using Eqs.~\eqref{eq:sigma-EAP} and \eqref{eq:Hamiltonian}, we have
\begin{align}
    \hat{H} \ket{\mathrm{EAP}}
    &= \sum_{X(\subset\Lambda)} \sum_{\vec{\mu}\in\{x,y,z\}^{X}}
    J_{X}^{\vec{\mu}} \bigotimes_{j \in X} \hat{\sigma}_{j}^{\mu_{j}} \ket{\mathrm{EAP}}\\
    &= \sum_{X(\subset\Lambda)} \sum_{\vec{\mu}\in\{x,y,z\}^{X}}
    J_{X}^{\vec{\mu}} \prod_{j \in X} \omega_{j}^{\mu_{j}} \bigotimes_{j \in X} \hat{\sigma}_{j+N/2}^{\mu_{j}} \ket{\mathrm{EAP}}.
\end{align}
Then, substituting Eq.~\eqref{eq:eigenstate-condition}, we obtain
\begin{align}
    \hat{H} \ket{\mathrm{EAP}}
    &= - \sum_{X(\subset\Lambda)} \sum_{\vec{\mu}\in\{x,y,z\}^{X}}
    J_{X+N/2}^{\vec{\mu}} \bigotimes_{j \in X} \hat{\sigma}_{j+N/2}^{\mu_{j}} \ket{\mathrm{EAP}}
    = - \hat{H} \ket{\mathrm{EAP}},
\end{align}
which is equivalent to statement~(\ref{statement:eigenstate-condition_2}). It obviously implies statement~(\ref{statement:eigenstate-condition_1}).
\end{proof}
Remark: The above proof also shows that, if we only need to show (\ref{statement:eigenstate-condition_3}) $\Rightarrow$ (\ref{statement:eigenstate-condition_2}), we can relax the condition on the locality of interactions in the main text, 
``$J_{X}^{\vec{\mu}}=0$ for all subset $X$ with $D(X)\ge N/4$'' to ``$J_{X}^{\vec{\mu}}=0$ for all subset $X$ with $D(X)\ge N/2$.''

\section{List of translation-invariant and nearest-neighbor-interacting Hamiltonians (Theorem~\ref{theorem:model})}
We provide a list of translation-invariant (for single-site translations) and nearest-neighbor-interacting Hamiltonians having an EAP state as an eigenstate. First, we exclude the case of free spins, as it is trivial. Next, for some cases where only one of $J^{\mu\nu}$ is non-zero, we find that $2^{N/2}$ EAP states are degenerate, so we also exclude such cases. Then pairs of the EAP state and the Hamiltonian are limited to the following five types (and their equivalents obtained by appropriate permutations of directions of the Pauli matrices). It can be readily confirmed through direct calculations that the pairs of the EAP state and the Hamiltonian listed below satisfy Eq.~\eqref{eq:eigenstate-condition} or \eqref{eq:Reduced_eigenstate-condition}.

\subsection{Case where the EAP state is invariant under $1$-site translation}
The EAP state $\ket{1;00}$ is an eigenstate of the Hamiltonian defined by
\begin{align}
    \hat{H} = \sum_{j=1}^{N} \left(
        J^{xy} \hat{\sigma}_{j}^{x} \hat{\sigma}_{j+1}^{y}
        + J^{yx} \hat{\sigma}_{j}^{y} \hat{\sigma}_{j+1}^{x}
        + J^{yz} \hat{\sigma}_{j}^{y} \hat{\sigma}_{j+1}^{z}
        + J^{zy} \hat{\sigma}_{j}^{z} \hat{\sigma}_{j+1}^{y}
        + h^{y} \hat{\sigma}_{j}^{y}
    \right),
\end{align}
for arbitrary values of $J^{xy},J^{yx},J^{yz},J^{zy},h^{y}$.
This model is Model~\ref{model:period1} in Theorem~\ref{theorem:model}.

\subsection{Case where the EAP state is invariant under $2$-site translation}
Suppose that $N/2$ is a multiple of two.

The EAP state $\ket{2;01,10}$ is an eigenstate of the Hamiltonian defined by
\begin{align}
    \hat{H} = \sum_{j=1}^{N} \left(
        J^{xx} \hat{\sigma}_{j}^{x} \hat{\sigma}_{j+1}^{x}
        + J^{zz} \hat{\sigma}_{j}^{z} \hat{\sigma}_{j+1}^{z}
    \right),
\end{align}
for arbitrary values of $J^{xx},J^{zz}$.
This model can be mapped onto free fermions via the Jordan-Wigner transformation.

In addition, the EAP state $\ket{2;10,11}$ is an eigenstate of the Hamiltonian defined by
\begin{align}
    \hat{H} = \sum_{j=1}^{N} \left(
        J^{xx} \hat{\sigma}_{j}^{x} \hat{\sigma}_{j+1}^{x}
        + J^{yy} \hat{\sigma}_{j}^{y} \hat{\sigma}_{j+1}^{y}
        + J^{xy} \hat{\sigma}_{j}^{x} \hat{\sigma}_{j+1}^{y}
        + J^{yx} \hat{\sigma}_{j}^{y} \hat{\sigma}_{j+1}^{x}
        + h^{z} \hat{\sigma}_{j}^{z}
    \right),
    \label{eq:H_period2_11}
\end{align}
for arbitrary values of $J^{xx},J^{yy},J^{xy},J^{yx}, h^{z}$.
This model can also be mapped onto free fermions via the Jordan-Wigner transformation.

\subsection{Case where the EAP state is invariant under $3$-site translation}
Suppose that $N/2$ is a multiple of three.

The EAP state $\ket{3;01,10,11}$ is an eigenstate of the Hamiltonian defined by
\begin{align}
    \hat{H} = \sum_{j=1}^{N} \left(
        J^{xy} \hat{\sigma}_{j}^{x} \hat{\sigma}_{j+1}^{y}
        + J^{yz} \hat{\sigma}_{j}^{y} \hat{\sigma}_{j+1}^{z}
    \right),
\end{align}
for arbitrary values of $J^{xy},J^{yz}$.
This model is Model~\ref{model:period3} in Theorem~\ref{theorem:model}. As shown in Theorem~\ref{theorem:nonintegrability}, this model is nonintegrable.

\subsection{Case where the EAP state is invariant under $4$-site translation}
Suppose that $N/2$ is a multiple of four.

The EAP state $\ket{4;00,01,10,11}$ is an eigenstate of the Hamiltonian defined by
\begin{align}
    \hat{H} = \sum_{j=1}^{N} \left(
        J^{xx} \hat{\sigma}_{j}^{x} \hat{\sigma}_{j+1}^{x}
        + J^{yz} \hat{\sigma}_{j}^{y} \hat{\sigma}_{j+1}^{z}
    \right)
\end{align}
for arbitrary values of $J^{xx},J^{yz}$.
This model is Model~\ref{model:period4} in Theorem~\ref{theorem:model}. As shown in Theorem~\ref{theorem:nonintegrability}, this model is nonintegrable.

\section{Proof of Theorem~\ref{theorem:nonintegrability}}

First we define a $k$-local conserved quantity (which is the same as one given in Ref.~\cite{Chiba2024}) by the operator $\hat{Q}$ that commutes with the Hamiltonian
\begin{align}
    [\hat{Q},\hat{H}]=0
    \label{eq:Commutation_Q_H}
\end{align}
and can be written as
\begin{align}
    \hat{Q} = \sum_{\ell=1}^{k}\sum_{\boldsymbol{A}^{\ell}}q_{\boldsymbol{A}^{\ell}_{j}}^{(\ell)}\hat{\boldsymbol{A}}^{\ell}_{j} + q_{I}^{(0)}\hat{I}.
    \label{eq:k-local}
\end{align}
Here $\boldsymbol{A}^{\ell}$ represents a sequence of symbols, $A^{1},A^{2},...,A^{\ell}$ satisfying
\begin{align}
    A^{1},A^{\ell}\in\{X,Y,Z\}
    \label{eq:Nonint_RangeA1}\\
    A^{2},...,A^{\ell-1}\in\{X,Y,Z,I\}
    \label{eq:Nonint_RangeA2}
\end{align}
and $\hat{\boldsymbol{A}}^{\ell}_{j}$ represents the product of the corresponding Pauli operators on the sites $\{j,j+1,...,j+\ell-1\}$:
\begin{align}
    \hat{\boldsymbol{A}}^{\ell}_{j}=\hat{A}^{1}_{j}\hat{A}^{2}_{j+1}...\hat{A}^{\ell}_{j+\ell-1}.
\end{align}
In Eq.~(\ref{eq:k-local}), $q_{\boldsymbol{A}^{\ell}_{j}}^{(\ell)}\in\mathbb{R}$ are the expansion coefficients. [We add the superscript $(\ell)$ in order to emphasize its value.] The crucial point of Eq.~(\ref{eq:k-local}) is that $\hat{Q}$ does not include $\hat{\boldsymbol{A}}^{\ell}_{j}$ with $\ell>k$.

Now we give the precise expression of Theorem~\ref{theorem:nonintegrability} of the main text, which is represented by the following two theorems:
\begin{theoremNonint} \label{theorem:Nonint_Period3}
    In Model~\ref{model:period3} with $J^{xy},J^{yz}\neq 0$, and for $k\le N/2$, 
    there is no $k$-local conserved quantity that is linearly independent of the Hamiltonian and the identity.
\end{theoremNonint}
\begin{theoremNonint} \label{theorem:Nonint_Period4}
    In Model~\ref{model:period4} with $J^{xx},J^{yz}\neq 0$, and for $k\le N/2$, 
    there is no $k$-local conserved quantity that is linearly independent of the Hamiltonian and the identity.
\end{theoremNonint}
In the remaining of this section, we prove these theorems by adapting the theoretical approach to prove the absence of local conserved quantities, which was introduced by N.~Shiraishi~\cite{Shiraishi2019}. There are only a few examples of such proofs~\cite{Shiraishi2019,Chiba2024,Park2024}. This approach starts from solving Eq.~(\ref{eq:Commutation_Q_H}) with respect to the coefficients with largest locality, $q_{\boldsymbol{A}^{k}_{j}}^{(k)}$, and showing that $q_{\boldsymbol{A}^{k}_{j}}^{(k)}=0$. When solving Eq.~(\ref{eq:Commutation_Q_H}), we need to calculate many commutators such as
\begin{align}
    &[q_{X_{j}IX}^{(3)}\hat{X}_{j}\hat{I}_{j+1}\hat{X}_{j+2},\sum_{j^{\prime}}J^{yz}\hat{Y}_{j^{\prime}}\hat{Z}_{j^{\prime}+1}] \nonumber\\
    &=q_{X_{j}IX}^{(3)}J^{yz}
    [\hat{X}_{j}\hat{I}_{j+1}\hat{X}_{j+2},
    \hat{Y}_{j-1}\hat{Z}_{j}+\hat{Y}_{j}\hat{Z}_{j+1}+\hat{Y}_{j+1}\hat{Z}_{j+2}+\hat{Y}_{j+2}\hat{Z}_{j+3}]\\
    &=2i q_{X_{j}IX}^{(3)}J^{yz}\bigl(
        -\hat{Y}_{j-1}\hat{Y}_{j}\hat{I}_{j+1}\hat{X}_{j+2}
        +\hat{Z}_{j}\hat{Z}_{j+1}\hat{X}_{j+2}
        -\hat{X}_{j}\hat{Y}_{j+1}\hat{Y}_{j+2}
        +\hat{X}_{j}\hat{I}_{j+1}\hat{Z}_{j+2}\hat{Z}_{j+3}
    \bigr).
    \label{eq:Nonint_Commutator_Example}
\end{align}
For simplicity of notation, we write $q_{X_{j}IX}^{(3)}$ in place of $q_{X_{j}I_{j+1}X_{j+2}}^{(3)}$. In order to express such calculations efficiently, we use the following diagrammatic notation:
\begin{align}
    \begin{array}{rlllllll}
         &       &X&I&X\\
         &Y      &Z& & \\ \hline
        -&Y_{j-1}&Y&I&X
    \end{array}
    \quad
    \begin{array}{rlllllll}
         &X    &I&X\\
         &Y    &Z& \\ \hline
         &Z_{j}&Z&X
    \end{array}
    \quad
    \begin{array}{rlllllll}
         &X    &I&X\\
         &     &Y&Z\\ \hline
        -&X_{j}&Y&Y
    \end{array}
    \quad
    \begin{array}{rlllllll}
         &X    &I&X& \\
         &     & &Y&Z\\ \hline
         &X_{j}&I&Z&Z
    \end{array}.
\end{align}
These four diagrams correspond to the four terms in Eq.~(\ref{eq:Nonint_Commutator_Example}). In each diagram, the first row represents the term from $\hat{Q}$, the second row the term from $\hat{H}$, and the third row the result of the commutator.
For simplicity of notation, we add the site index only for the leftmost operators in the third row. In addition, we call the first row of the diagram ``$\ell$-local input'', and the third row of the diagram ``$\ell$-local output'', when they consist of $X$, $Y$, $Z$, and $I$ on $\ell$ consecutive sites.

\subsection{Proof of Theorem~\ref{theorem:Nonint_Period3}}
This subsection proves Theorem~\ref{theorem:Nonint_Period3}.
Throughout this subsection, we consider Model~\ref{model:period3} and assume $J^{xy},J^{yz}\neq 0$ and $k\le N/2$.
(The reason for the assumption $k\le N/2$ is the same as one discussed in Sec.~VI~A of Ref.~\cite{Chiba2024}.)
\begin{proof}
The proof of Theorem~\ref{theorem:Nonint_Period3} is divided into three parts.
The first part investigates the coefficients with largest locality, $q_{\boldsymbol{A}^{k}_{j}}^{(k)}$. For the coefficients of the form $q_{Z_{j}A^{2}...A^{k-1}X}^{(k)}$, we have
\begin{align}
\begin{array}{rlllllll}
 &Z      &A^{2}&...&A^{k-1}&X& \\
 & & &...& &Y&Z \\ \hline
 &Z_{j}&A^{2}&...&A^{k-1}&Z&Z
\end{array}.
\label{eq:Nonint_Period3_Z...ZZ(k+1)}
\end{align}
Because a $(k+1)$-local output can be obtained only when the Hamiltonian term is applied to the edges of $k$-local inputs, there are at most two $k$-local inputs that contribute to one $(k+1)$-local output.
However, since the left end of the output of Eq.~(\ref{eq:Nonint_Period3_Z...ZZ(k+1)}) is $Z$, the other contribution does not exist.
Furthermore, from Eq.~(\ref{eq:Commutation_Q_H}), the sum of all contribution to the output $Z_{j}A^{2}...A^{k-1}ZZ$ must vanish, and therefore we have
\begin{align}
    J^{yz}q_{Z_{j}A^{2}...A^{k-1}X}^{(k)}=0.
\end{align}
In a similar manner, we can obtain the following lemma:
\begin{lemma}\label{lemma:Nonint_period3_Z...(k)=0}
    For $2\le k\le N/2$, the solution of Eq.~(\ref{eq:Commutation_Q_H}) satisfies
    \begin{align}
        q_{Z_{j}A^{2}...A^{k-1}A^{k}}^{(k)}&=0\\
        q_{A^{1}_{j}A^{2}...A^{k-1}X}^{(k)}&=0
    \end{align}
    for all $j\in\Lambda$.
    Here the symbols $A^{1},A^{2},...,A^{k-1},A^{k}$ that are not specified can be any symbols satisfying Eqs.~(\ref{eq:Nonint_RangeA1}) and (\ref{eq:Nonint_RangeA2}).
\end{lemma}

Next we examine the coefficients of the form $q_{X_{j}A^{2}...A^{k-1}Y}^{(k)}$.
They have the following contributions
\begin{align}
    \begin{array}{rlllllll}
         &X&A^{2}&A^{3}&...&A^{k-1}&Y& \\
         & & & &...& &X&Y \\ \hline
        -&X_{j}&A^{2}&A^{3}&...&A^{k-1}&Z&Y
    \end{array}.
    \label{eq:Nonint_Period3_X...ZY(k+1)}
\end{align}
When $A^{2}=I,Y$, this $(k+1)$-local output does not have the other contribution.
When $A^{2}=X$, it has the other contribution,
\begin{align}
    \begin{array}{rlllllll}
         & &Z&A^{3}&...&A^{k-1}&Z&Y \\
         &X&Y& &...& & &  \\ \hline
        -&X_{j}&X&A^{3}&...&A^{k-1}&Z&Y
    \end{array},
\end{align}
which however vanishes from Lemma~\ref{lemma:Nonint_period3_Z...(k)=0}.
In a similar manner, we can obtain the following lemma:
\begin{lemma}\label{lemma:Nonint_period3_XX...(k)=0}
    For $3\le k\le N/2$, the solution of Eq.~(\ref{eq:Commutation_Q_H}) satisfies
    \begin{align}
        q_{X_{j}A^{2}...A^{k-1}A^{k}}^{(k)}&=0\quad \text{for }A^{2}=I,Y,X\\
        q_{Y_{j}A^{2}...A^{k-1}A^{k}}^{(k)}&=0\quad \text{for }A^{2}=I,Z\\
        q_{A^{1}_{j}A^{2}...A^{k-1}Z}^{(k)}&=0\quad \text{for }A^{k-1}=I,Y,Z\\
        q_{A^{1}_{j}A^{2}...A^{k-1}Y}^{(k)}&=0\quad \text{for }A^{k-1}=I,X
    \end{align}
    for all $j\in\Lambda$.
    Here the symbols $A^{1},A^{2},...,A^{k-1},A^{k}$ that are not specified can be any symbols satisfying Eqs.~(\ref{eq:Nonint_RangeA1}) and (\ref{eq:Nonint_RangeA2}).
\end{lemma}

Furthermore, we can obtain relation between two of the remaining $k$-local inputs as in
\begin{align}
    \begin{array}{rlllllll}
         &X&Z&A^{3}&...&A^{k-1}&Y& \\
         & & & &...& &X&Y \\ \hline
        -&X_{j}&Z&A^{3}&...&A^{k-1}&Z&Y
    \end{array}
    \quad
    \begin{array}{rlllllll}
         & &X&A^{3}&...&A^{k-1}&Z&Y \\
         &X&Y& &...& & &  \\ \hline
         &X_{j}&Z&A^{3}&...&A^{k-1}&Z&Y
    \end{array},
\end{align}
which results in
\begin{align}
    -J^{xy}q_{X_{j}ZA^{3}...A^{k-1}Y}^{(k)}
    +J^{xy}q_{X_{j+1}A^{3}...A^{k-1}ZY}^{(k)}
    =0.
\end{align}
If $A^{3}$ is not $Z$, we have $q_{X_{j+1}A^{3}...A^{k-1}ZY}^{(k)}=0$ from Lemma~\ref{lemma:Nonint_period3_XX...(k)=0}, and hence $q_{X_{j}ZA^{3}...A^{k-1}Y}^{(k)}=0$.
By using such a relation, we can shift the symbols $A^{3},...,A^{k-1}$ to the left and we can determine these symbols.
As a result, we can obtain the following proposition:
\begin{proposition}\label{proposition:Nonint_Period3_Step1}
    For any $j\in\Lambda$ and for any $\boldsymbol{A}^{k}$, the solution of Eq.~(\ref{eq:Commutation_Q_H}) satisfies
    \begin{align}
        q_{\boldsymbol{A}^{k}_{j}}^{(k)}=0,
    \end{align}
    except for 
    \begin{align}
        q_{X_{j}(Z)^{k-2}Y}^{(k)},\quad 
        q_{Y_{j}(X)^{k-2}Z}^{(k)},\quad
        q_{Y_{j}(X)^{n}Y(Z)^{k-n-3}Y}^{(k)}\quad \text{with }n=0,1,...,k-3.
    \end{align}
    In addition, these remaining coefficients are independent of the site $j$ and satisfy
    \begin{align}
        \frac{q_{Y_{j}(X)^{n}Y(Z)^{k-n-3}Y}^{(k)}}{(J^{xy})^{k-n-2}(J^{yz})^{n+1}}
        =-\frac{q_{Y_{j}(X)^{k-2}Z}^{(k)}}{(J^{yz})^{k-1}}
        =-\frac{q_{X_{j}(Z)^{k-2}Y}^{(k)}}{(J^{xy})^{k-1}}
        =-\frac{q_{X_{1}(Z)^{k-2}Y}^{(k)}}{(J^{xy})^{k-1}}\quad \text{for }n=0,1,...,k-3.
    \end{align}
    Here we used a shorthand notation of a sequence of symbols
    \begin{align}
        (A)^{m}:=\underbrace{AA...A}_{m\, \mathrm{times}}\quad \text{ for }A=X,Y,Z,I.
    \end{align}
\end{proposition}
Therefore we only need to show that one of these remaining coefficients is zero.

As the second part of the proof, we examine the coefficient $q_{X_{j}(Z)^{k-2}Y}^{(k)}(=q_{X_{1}(Z)^{k-2}Y}^{(k)})$.
We consider the contribution from $q_{X_{j}(Z)^{k-2}Y}^{(k)}$ to a $k$-local output which can also include the contribution from $(k-1)$-local inputs.
For instance,
\begin{align}
    \begin{array}{rlllllll}
         &X&(Z)^{k-4}&Z&Z&Y \\
         & & &Y&Z& \\ \hline
        -&X_{j}&(Z)^{k-4}&X&I&Y
    \end{array}
    \quad
    \begin{array}{rlllllll}
         & &X&(Z)^{k-5}&X&I&Y \\
         &X&Y& & & & \\ \hline
         &X_{j}&Z&(Z)^{k-5}&X&I&Y
    \end{array}
\end{align}
are the only contribution to the $k$-local output $X_{j}(Z)^{k-4}XIY$, and hence we have
\begin{align}
    J^{xy}q_{X_{j+1}(Z)^{k-5}XIY}^{(k-1)}
    =J^{yz}q_{X_{1}(Z)^{k-2}Y}^{(k)}.
    \label{eq:Nonint_period3_Prop2_1}
\end{align}
For coefficients of $(k-1)$-local inputs, we obtain the following:
All contributions to $k$-local output $X_{j}(Z)^{k-n-5}XI(Z)^{n+1}Y$ (for $n=0,...,k-6$) are given by
\begin{gather}
    \begin{array}{rllllllllll}
         &X&(Z)^{k-n-5}&X&I&(Z)^{n}&Y& \\
         & & & & & &X&Y\\ \hline
        -&X_{j}&(Z)^{k-n-5}&X&I&(Z)^{n}&Z&Y
    \end{array}
    \quad
    \begin{array}{rllllllllll}
         & &X&(Z)^{k-n-6}&X&I&(Z)^{n+1}&Y \\
         &X&Y& & & & & \\ \hline
         &X_{j}&Z&(Z)^{k-n-5}&X&I&(Z)^{n+1}&Y
    \end{array}
    \nonumber\\
    \begin{array}{rllllllllll}
         &X&(Z)^{k-n-5}&Z&Z&(Z)^{n+1}&Y \\
         & & &Y&Z& & \\ \hline
        -&X_{j}&(Z)^{k-n-5}&X&I&(Z)^{n+1}&Y
    \end{array}
    \quad
    \begin{array}{rllllllllll}
         &X&(Z)^{k-n-5}&X&X&X&(Z)^{n}&Y \\
         & & & &X&Y& & \\ \hline
         &X_{j}&(Z)^{k-n-5}&X&I&Z&(Z)^{n}&Y
    \end{array},
\end{gather}
which result in 
\begin{align}
    J^{xy}\bigl(
    q_{X_{j+1}(Z)^{k-n-6}XI(Z)^{n+1}Y}^{(k-1)}
    -q_{X_{j}(Z)^{k-n-5}XI(Z)^{n}Y}^{(k-1)}
    \bigr)
    =J^{yz}q_{X_{1}(Z)^{k-2}Y}^{(k)}
    \quad\text{for all }n=0,...,k-6.
    \label{eq:Nonint_period3_Prop2_2}
\end{align}
For the coefficient $q_{X_{j}XI(Z)^{k-5}Y}^{(k-1)}$, which appears in $n=k-6$ case of the above equation, we can obtain another relation
\begin{align}
    J^{xy}\bigl(
    -q_{Z_{j+1}I(Z)^{k-4}Y}^{(k-1)}
    -q_{X_{j}XI(Z)^{k-5}Y}^{(k-1)}
    \bigr)
    =J^{yz}q_{X_{1}(Z)^{k-2}Y}^{(k)},
    \label{eq:Nonint_period3_Prop2_3}
\end{align}
in a similar manner.
We can also obtain the following relation for the coefficient $q_{Z_{j}I(Z)^{k-4}Y}^{(k-1)}$,
\begin{align}
    J^{xy}q_{Z_{j}I(Z)^{k-4}Y}^{(k-1)}
    =2J^{yz}q_{X_{1}(Z)^{k-2}Y}^{(k)},
    \label{eq:Nonint_period3_Prop2_4}
\end{align}
by considering the contributions to $Z_{j}I(Z)^{k-3}Y$,
\begin{gather}
    \begin{array}{rllllllllll}
         &Z&I&(Z)^{k-4}&Y& \\
         & & & &X&Y \\ \hline
        -&Z_{j}&I&(Z)^{k-4}&Z&Y
    \end{array}
    \quad
    \begin{array}{rllllllllll}
         &Z&X&X&(Z)^{k-4}&Y \\
         & &X&Y& & \\ \hline
         &Z_{j}&I&Z&(Z)^{k-4}&Y
    \end{array}
    \quad
        \begin{array}{rllllllllll}
         &Y&Y&(Z)^{k-3}&Y \\
         &X&Y& & \\ \hline
        -&Z_{j}&I&(Z)^{k-3}&Y
    \end{array}
    \quad
    \begin{array}{rllllllllll}
         &X&Z&(Z)^{k-3}&Y \\
         &Y&Z& & \\ \hline
         &Z_{j}&I&(Z)^{k-3}&Y
    \end{array}.
\end{gather}
Because the sum of the left-hand sides of Eqs.~(\ref{eq:Nonint_period3_Prop2_1}), (\ref{eq:Nonint_period3_Prop2_2})--(\ref{eq:Nonint_period3_Prop2_4}) (by choosing the site $j$ appropriately) become zero, we have
\begin{align}
    0&=
    J^{xy}q_{X_{1}(Z)^{k-5}XIY}^{(k-1)}
    +\sum_{n=0}^{k-6}J^{xy}\bigl(
    q_{X_{n+2}(Z)^{k-n-6}XI(Z)^{n+1}Y}^{(k-1)}
    -q_{X_{n+1}(Z)^{k-n-5}XI(Z)^{n}Y}^{(k-1)}
    \bigr)
    \nonumber\\
    &\hspace{10pt}+J^{xy}\bigl(
    -q_{Z_{k-3}I(Z)^{k-4}Y}^{(k-1)}
    -q_{X_{k-4}XI(Z)^{k-5}Y}^{(k-1)}
    \bigr)
    +J^{xy}q_{Z_{k-3}I(Z)^{k-4}Y}^{(k-1)}\\
    &=(k-1)J^{yz}q_{X_{1}(Z)^{k-2}Y}^{(k)}.
\end{align}
Thus we obtain the following proposition:
\begin{proposition}\label{proposition:Nonint_Period3_Step2}
    For $3\le k\le N/2$, the solution of Eq.~(\ref{eq:Commutation_Q_H}) satisfies
    \begin{align}
        q_{X_{1}(Z)^{k-2}Y}^{(k)}=0.
    \end{align}
\end{proposition}
By combining Propositions~\ref{proposition:Nonint_Period3_Step1} and \ref{proposition:Nonint_Period3_Step2}, we have
\begin{align}
    q_{\boldsymbol{A}^{k}_{j}}^{(k)}=0\quad \text{for all }j\text{ and }\boldsymbol{A}^{k}.
    \label{eq:Nonint_Period3_q^k=0}
\end{align}
This means that $\hat{Q}$ is a $(k-1)$-local conserved quantity.
Applying the same argument to $k-1$, $k-2$,..., and $3$-local conserved quantity, we can show that any $k$-local conserved quantity with $k\le N/2$ have to be a $2$-local conserved quantity.

As the third part of the proof, we analyze the coefficients $q_{\boldsymbol{A}^{\ell}_{j}}^{(\ell)}$ with $\ell\le k$ in the case of $k\le 2$.
From Lemma~\ref{lemma:Nonint_period3_Z...(k)=0}, we only need to consider the coefficients of the form $q_{X_{j}Y}^{(2)},q_{X_{j}Z}^{(2)},q_{Y_{j}Y}^{(2)},q_{Y_{j}Z}^{(2)}$ and $q_{A_{j}}^{(1)}$.
Furthermore, the coefficient $q_{X_{j}Z}^{(2)}$ vanishes because
\begin{align}
    \begin{array}{rlllllll}
         &X&Z&  \\
         & &X&Y \\ \hline
        &X_{j}&Y&Y
    \end{array}
\end{align}
is the only contribution to the $3$-local output $X_{j}YY$.
The coefficient $q_{Y_{j}Y}^{(2)}$ vanishes by a similar reason.
In addition, because $2$-local output comes only from $1$-local input,
we can easily show that $q_{X_{j}}^{(1)}=q_{Y_{j}}^{(1)}=q_{Z_{j}}^{(1)}=0$.
For the remaining coefficients $q_{X_{j}Y}^{(2)}$ and $q_{Y_{j}Z}^{(2)}$, we can easily show that they are independent of the site $j$ and are related to each other by
\begin{align}
    J^{xy}q_{Y_{j}Z}^{(2)}-J^{yz}q_{X_{j+1}Y}^{(2)}=0,
\end{align}
which results in the following proposition:
\begin{proposition}\label{proposition:Nonint_Period3_Step3}
    Any $2$-local conserved quantity $\hat{Q}$ can be written as
    \begin{align}
        \hat{Q}=a\hat{H}+b\hat{I},
    \end{align}
    with arbitrary constants $a,b\in\mathbb{R}$.
\end{proposition}
From Eq.~(\ref{eq:Nonint_Period3_q^k=0}) and Proposition~\ref{proposition:Nonint_Period3_Step3}, we obtain Theorem~\ref{theorem:Nonint_Period3}.
\end{proof}

\subsection{Proof of Theorem~\ref{theorem:Nonint_Period4}}
This subsection proves Theorem~\ref{theorem:Nonint_Period4}.
Throughout this subsection, we consider Model~\ref{model:period4} and assume $J^{xx},J^{yz}\neq 0$ and $k\le N/2$.
(The reason for the assumption $k\le N/2$ is the same as one discussed in Sec.~VI~A of Ref.~\cite{Chiba2024}.)
\begin{proof}
The proof of Theorem~\ref{theorem:Nonint_Period3} is divided into three parts.
The first part investigates the coefficients with largest locality, $q_{\boldsymbol{A}^{k}_{j}}^{(k)}$.
In a manner similar to the proof of Lemma~\ref{lemma:Nonint_period3_Z...(k)=0}, we can show the following lemma:
\begin{lemma}\label{lemma:Nonint_period4_Z...(k)=0}
    For $2\le k\le N/2$, the solution of Eq.~(\ref{eq:Commutation_Q_H}) satisfies
    \begin{align}
        q_{Z_{j}A^{2}...A^{k-1}A^{k}}^{(k)}&=0\\
        q_{A^{1}_{j}A^{2}...A^{k-1}Y}^{(k)}&=0
    \end{align}
    for all $j\in\Lambda$.
    Here the symbols $A^{1},A^{2},...,A^{k-1},A^{k}$ that are not specified can be any symbols satisfying Eqs.~(\ref{eq:Nonint_RangeA1}) and (\ref{eq:Nonint_RangeA2}).
\end{lemma}

Furthermore, we can show the following lemma in a manner similar to the proof of Lemma~\ref{lemma:Nonint_period3_XX...(k)=0}:
\begin{lemma}\label{lemma:Nonint_period4_XX...(k)=0}
    For $3\le k\le N/2$, the solution of Eq.~(\ref{eq:Commutation_Q_H}) satisfies
    \begin{align}
        q_{X_{j}A^{2}...A^{k-1}A^{k}}^{(k)}&=0\quad \text{for }A^{2}=I,Y,X\\
        q_{Y_{j}A^{2}...A^{k-1}A^{k}}^{(k)}&=0\quad \text{for }A^{2}=I,Z\\
        q_{A^{1}_{j}A^{2}...A^{k-1}X}^{(k)}&=0\quad \text{for }A^{k-1}=I,X,Z\\
        q_{A^{1}_{j}A^{2}...A^{k-1}Z}^{(k)}&=0\quad \text{for }A^{k-1}=I,Y
    \end{align}
    for all $j\in\Lambda$.
    Here the symbols $A^{1},A^{2},...,A^{k-1},A^{k}$ that are not specified can be any symbols satisfying Eqs.~(\ref{eq:Nonint_RangeA1}) and (\ref{eq:Nonint_RangeA2}).
\end{lemma}
By shifting the symbols $A^{3},...,A^{k-1}$ to the left as we did to obtain Proposition~\ref{proposition:Nonint_Period3_Step1}, many coefficients can be shown to be zero.
To explain the result, we introduce a version of ``doubling product.''
It was originally introduced by N.~Shiraishi~\cite{Shiraishi2019}.
Our version is modified for analyzing Model~\ref{model:period4} as follows:
We call a sequence of the Pauli operators doubling product, if it can by written as $\overline{A^{1}_{j}A^{2}...A^{n}}$, where
\begin{align}
    \overline{A^{1}_{j}}
    &:= \begin{cases}
        X_{j}X_{j+1}\quad&\text{when }A^{1}=X\\
        Y_{j}Z_{j+1}\quad&\text{when }A^{1}=Y
    \end{cases}
    \label{eq:Period4_DoublingProduct1}\\
    \overline{A^{1}_{j}A^{2}...A^{n}A^{n+1}}
    &:= \begin{cases}
        c\overline{A^{1}_{j}A^{2}...A^{n}}\times X_{j+n}X_{j+n+1}
        \quad&\text{when }A^{n+1}=X,\ A^{n}\neq X\\
        c\overline{A^{1}_{j}A^{2}...A^{n}}\times Y_{j+n}Z_{j+n+1}
        \quad&\text{when }A^{n+1}=Y
    \end{cases}.
    \label{eq:Period4_DoublingProduct2}
\end{align}
Here $c$ is chosen from $\{\pm i\}$ to make its coefficient $1$.
In addition, we introduce $J_{\overline{A^{1}...A^{n}}}$ by
\begin{align}
    J_{\overline{A^{1}}}
    &:= \begin{cases}
        J^{xx}\quad\text{when }A^{1}=X,\\
        J^{yz}\quad\text{when }A^{1}=Y
    \end{cases}\\
    J_{\overline{A^{1}...A^{n}A^{n+1}}}
    &:= \begin{cases}
        J_{\overline{A^{1}...A^{n}}}
        \times J^{xx}
        \quad&\text{when }A^{n}=Y,\ A^{n+1}=X\\
        J_{\overline{A^{1}...A^{n}}}
        \times J^{yz}
        \quad&\text{when }A^{n}=X,\ A^{n+1}=Y\\
       -J_{\overline{A^{1}...A^{n}}}
        \times J^{yz}
        \quad&\text{when }A^{n}=Y,\ A^{n+1}=Y
    \end{cases}.
\end{align}
Then we can obtain the following proposition:
\begin{proposition}\label{proposition:Nonint_Period4_Step1}
    For any $j\in\Lambda$ and for any $\boldsymbol{A}^{k}$ other than doubling product, the solution of Eq.~(\ref{eq:Commutation_Q_H}) satisfies
    \begin{align}
        q_{\boldsymbol{A}^{k}_{j}}^{(k)}=0.
    \end{align}
    For the case where $\boldsymbol{A}^{k}$ is given by a doubling product Eqs.~(\ref{eq:Period4_DoublingProduct1}) and (\ref{eq:Period4_DoublingProduct2}), these remaining coefficients are independent of the site $j$ and are related to each other by
    \begin{align}
        \frac{q_{\overline{A^{1}_{j}...A^{k-1}}}^{(k)}}{J_{\overline{A^{1}...A^{k-1}}}}
        =\frac{q_{Y_{j}(X)^{k-2}Z}^{(k)}}{J^{yz}(-J^{yz})^{k-2}}
        =\frac{q_{Y_{1}(X)^{k-2}Z}^{(k)}}{J^{yz}(-J^{yz})^{k-2}}
        \label{eq:Nonint_Period4_Step1_q/J}
    \end{align}
    for any $j\in\Lambda$.
\end{proposition}
Therefore we only need to show that one of these remaining coefficients is zero.

As the second part of the proof, we examine the coefficient $q_{Y_{1}(X)^{k-2}Z}^{(k)}$.
We consider the contribution from the $k$-local input $Y_{j}(X)^{k-2}Z$ to a $k$-local output which can also include the contribution from $(k-1)$-local inputs. For instance,
\begin{align}
    \begin{array}{rlllllll}
         &Y&X&(X)^{k-3}&Z \\
         &X&X& & \\ \hline
        -&Z_{j}&I&(X)^{k-3}&Z
    \end{array}
    \quad
    \begin{array}{rlllllll}
         &X&Z&(X)^{k-3}&Z \\
         &Y&Z& & \\ \hline
         &Z_{j}&I&(X)^{k-3}&Z
    \end{array}
    \quad
    \begin{array}{rlllllll}
         &Z&Y&Y&(X)^{k-4}&Z \\
         & &Y&Z& & \\ \hline
         &Z_{j}&I&X&(X)^{k-4}&Z
    \end{array}
    \quad
    \begin{array}{rlllllll}
         &Z&I&(X)^{k-4}&Z& \\
         & & & &Y&Z\\ \hline
        -&Z_{j}&I&(X)^{k-4}&X&Z
    \end{array}
\end{align}
are the only contribution to the $k$-local output $Z_{j}I(X)^{k-3}Z$.
Note that the contribution from the third diagram vanishes because of Lemma~\ref{lemma:Nonint_period4_Z...(k)=0}, and the contribution from the second diagram satisfies
\begin{align}
    q_{X_{j}Z(X)^{k-3}Z}^{(k)}=-\frac{J^{xx}}{J^{yz}}q_{Y_{1}(X)^{k-2}Z}^{(k)}
\end{align}
because $XZ(X)^{k-3}Z$ can be written as a doubling product $\overline{X(Y)^{k-2}}$ and Eq.~(\ref{eq:Nonint_Period4_Step1_q/J}) in Proposition~\ref{proposition:Nonint_Period4_Step1} is applicable.
Hence we have
\begin{align}
    J^{yz}q_{Z_{j}I(X)^{k-4}Z}^{(k-1)}
    =-2J^{xx}q_{Y_{1}(X)^{k-2}Z}^{(k)}.
    \label{eq:Nonint_period4_Prop2_1}
\end{align}
In a similar manner, we can obtain
\begin{align}
    J^{yz}q_{X_{j}YI(X)^{k-5}Z}^{(k-1)}
    -J^{xx}q_{Z_{j+1}I(X)^{k-4}Z}^{(k-1)}
    &=-\frac{(J^{xx})^2}{J^{yz}}q_{Y_{1}(X)^{k-2}Z}^{(k)}
    \label{eq:Nonint_period4_Prop2_2}\\
    -J^{yz}q_{Y_{j}YYI(X)^{k-6}Z}^{(k-1)}
    -J^{yz}q_{X_{j+1}YI(X)^{k-5}Z}^{(k-1)}
    &=-\frac{(J^{xx})^2}{J^{yz}}q_{Y_{1}(X)^{k-2}Z}^{(k)}
    \label{eq:Nonint_period4_Prop2_3}\\
    -J^{yz}q_{Y_{j}(X)^{n+1}YYI(X)^{k-7-n}Z}^{(k-1)}
    +J^{yz}q_{Y_{j+1}(X)^{n}YYI(X)^{k-6-n}Z}^{(k-1)}
    &=-\frac{(J^{xx})^2}{J^{yz}}q_{Y_{1}(X)^{k-2}Z}^{(k)}\quad\text{for }n=0,...,k-7
    \label{eq:Nonint_period4_Prop2_4}\\
    J^{yz}q_{Y_{j}(X)^{k-6}YYIZ}^{(k-1)}
    &=-\frac{(J^{xx})^2}{J^{yz}}q_{Y_{1}(X)^{k-2}Z}^{(k)}
    \label{eq:Nonint_period4_Prop2_5}
\end{align}
from the diagrams 
\begin{gather}
    \begin{array}{rlllllll}
     & &Z&I&(X)^{k-4}&Z \\
     &X&X& & & \\ \hline
    -&X_{j}&Y&I&(X)^{k-4}&Z
    \end{array}
    \quad
    \begin{array}{rlllllll}
     &X&Y&I&(X)^{k-5}&Z& \\
     & & & & &Y&Z \\ \hline
    -&X_{j}&Y&I&(X)^{k-5}&X&Z
    \end{array}
    \quad
    \begin{array}{rlllllll}
     &X&Z&X&(X)^{k-4}&Z \\
     & &X&X& & \\ \hline
     &X_{j}&Y&I&(X)^{k-4}&Z
    \end{array}
    \quad
    \begin{array}{rlllllll}
     &X&Y&Y&Y&(X)^{k-5}&Z& \\
     & & &Y&Z& & \\ \hline
    -&X_{j}&Y&I&X&(X)^{k-5}&Z
    \end{array},
\end{gather}
\begin{gather}
    \begin{array}{rlllllll}
     & &X&Y&I&(X)^{k-5}&Z \\
     &Y&Z& & & \\ \hline
    -&Y_{j}&Y&Y&I&(X)^{k-5}&Z
    \end{array}
    \quad
    \begin{array}{rlllllll}
     &Y&Y&Y&I&(X)^{k-6}&Z& \\
     & & & & & &Y&Z \\ \hline
    -&Y_{j}&Y&Y&I&(X)^{k-6}&X&Z
    \end{array}
    \quad
    \begin{array}{rlllllll}
     &Y&Y&Z&X&(X)^{k-5}&Z \\
     & & &X&X& & \\ \hline
     &Y_{j}&Y&Y&I&(X)^{k-5}&Z
    \end{array}
    \quad
    \begin{array}{rlllllll}
     &Y&Y&Y&Y&Y&(X)^{k-6}&Z\\
     & & & &Y&Z& & \\ \hline
     &Y_{j}&Y&Y&I&X&(X)^{k-6}&Z
    \end{array},
\end{gather}
\begin{gather}
    \begin{array}{rlllllllllll}
     & &Y&(X)^{n}&Y&Y&I&(X)^{k-6-n}&Z \\
     &Y&Z& & & & & & \\ \hline
     &Y_{j}&X&(X)^{n}&Y&Y&I&(X)^{k-6-n}&Z
    \end{array}
    \quad
    \begin{array}{rlllllllllll}
     &Y&(X)^{n+1}&Y&Y&I&(X)^{k-7-n}&Z& \\
     & & & & & & &Y&Z\\ \hline
     &Y_{j}&(X)^{n+1}&Y&Y&I&(X)^{k-7-n}&X&Z
    \end{array}
    \nonumber\\[5pt]
    \begin{array}{rlllllllllll}
     &Y&(X)^{n+1}&Y&Z&X&(X)^{k-6-n}&Z \\
     & & & &X&X& & \\ \hline
     &Y_{j}&(X)^{n+1}&Y&Y&I&(X)^{k-6-n}&Z
    \end{array}
    \quad
    \begin{array}{rlllllllllll}
     &Y&(X)^{n+1}&Y&Y&Y&Y&(X)^{k-7-n}&Z\\
     & & & & &Y&Z& & \\ \hline
     &Y_{j}&(X)^{n+1}&Y&Y&I&X&(X)^{k-7-n}&Z
    \end{array},
\end{gather}
\begin{gather}
    \begin{array}{rlllllll}
     & &Y&(X)^{k-6}&Y&Y&I&Z \\
     &Y&Z& & & & & \\ \hline
     &Y_{j}&X&(X)^{k-6}&Y&Y&I&Z
    \end{array}
    \quad
    \begin{array}{rlllllll}
     &Y&(X)^{k-5}&Y&Z&X&Z \\
     & & & &X&X& \\ \hline
     &Y_{j}&(X)^{k-5}&Y&Y&I&Z
    \end{array}
    \quad
    \begin{array}{rlllllll}
     &Y&(X)^{k-6}&Y&Y&X&Y\\
     & & & & &X&X\\ \hline
    -&Y_{j}&(X)^{k-6}&Y&Y&I&Z
    \end{array},
\end{gather}
respectively.
Because the sum of the left-hand sides of Eq.~(\ref{eq:Nonint_period4_Prop2_1}) $\times J^{xx}/J^{yz}$ and of Eqs.~(\ref{eq:Nonint_period4_Prop2_2})--(\ref{eq:Nonint_period4_Prop2_5}) (by choosing the site $j$ appropriately) become zero, we have
\begin{align}
    0
    &= J^{yz}q_{Y_{1}(X)^{k-6}YYIZ}^{(k-1)}
    +\sum_{n=0}^{k-7}J^{yz}\bigl(
    -q_{Y_{n+1}(X)^{k-n-6}YYI(X)^{n}Z}^{(k-1)}
    +q_{Y_{n+2}(X)^{k-n-7}YYI(X)^{n+1}Z}^{(k-1)}
    \bigr)
    \nonumber\\
    &\hspace{10pt}+J^{yz}\bigl(
    -q_{Y_{k-5}YYI(X)^{k-6}Z}^{(k-1)}
    -q_{X_{k-4}YI(X)^{k-5}Z}^{(k-1)}
    \bigr)
    +\bigl(J^{yz}q_{X_{k-4}YI(X)^{k-5}Z}^{(k-1)}
    -J^{xx}q_{Z_{k-3}I(X)^{k-4}Z}^{(k-1)}\bigr)
    \nonumber\\
    &\hspace{10pt}+J^{xx}q_{Z_{k-3}I(X)^{k-4}Z}^{(k-1)}\\
    &=-(k-1)\frac{(J^{xx})^2}{J^{yz}}q_{Y_{1}(X)^{k-2}Z}^{(k)}.
\end{align}
Thus we obtain the following proposition:
\begin{proposition}\label{proposition:Nonint_Period4_Step2}
    For $3\le k\le N/2$, the solution of Eq.~(\ref{eq:Commutation_Q_H}) satisfies
    \begin{align}
        q_{Y_{1}(X)^{k-2}Z}^{(k)}=0.
    \end{align}
\end{proposition}
By combining Propositions~\ref{proposition:Nonint_Period3_Step1} and \ref{proposition:Nonint_Period3_Step2}, we have
\begin{align}
    q_{\boldsymbol{A}^{k}_{j}}^{(k)}=0
    \quad \text{for all }j\text{ and }\boldsymbol{A}^{k}.
    \label{eq:Nonint_Period4_q^k=0}
\end{align}
This means that $\hat{Q}$ is a $(k-1)$-local conserved quantity.
Applying the same argument to $k-1$, $k-2$,..., and $3$-local conserved quantity, we can show that any $k$-local conserved quantity with $k\le N/2$ have to be a $2$-local conserved quantity.

As the third part of the proof, it is straightforward to show the following proposition:
\begin{proposition}\label{proposition:Nonint_Period4_Step3}
    Any $2$-local conserved quantity $\hat{Q}$ can be written as
    \begin{align}
        \hat{Q}=a\hat{H}+b\hat{I},
    \end{align}
    with arbitrary constants $a,b\in\mathbb{R}$.
\end{proposition}
From Eq.~(\ref{eq:Nonint_Period4_q^k=0}) and Proposition~\ref{proposition:Nonint_Period4_Step3}, we obtain Theorem~\ref{theorem:Nonint_Period4}.
\end{proof}

\section{Level spacing statistics of Model~\ref{model:period1}}
In the paragraph below Theorem~\ref{theorem:nonintegrability} of the main text, we explained that Model~\ref{model:period1} is expected to be nonintegrable.
To confirm this expectation, we numerically investigate the level spacing statistics~\cite{Mehta2004,Atas2013} of Model~\ref{model:period1} by exact diagonalization. Figure~\ref{fig:LevelSpacing_model1} plots the distribution of the ratio of consecutive level spacings~\cite{Atas2013}
\begin{align}
    r = \min \Bigl\{\frac{E_{n+2}^{q}-E_{n+1}^{q}}{E_{n+1}^{q}-E_{n}^{q}}, \frac{E_{n+1}^{q}-E_{n}^{q}}{E_{n+2}^{q}-E_{n+1}^{q}}\Bigr\},
\end{align}
constructed from eigenenergy $E_{n}^{q}$ in the eigenspace of translation with momentum $q=2\pi/N$~\footnote{To remove the possibility of accidental discrete symmetries, we avoid eigenspace of special momentum such as $0$ and $\pi$.}
(sorted in descending order). We set the parameters as $J^{xy}=e,J^{yx}=1,J^{yz}=\pi,J^{zy}=0,h^{y}=\ln 7$~\footnote{We can always take $J^{zy}=0$ by an appropriate rotation around $y$ axis.}. The plot is well described by the Gaussian unitary ensemble distribution~\cite{Atas2013} (dashed line) and well separated from the Poisson distribution~\cite{Atas2013} (dotted line). This shows that Model~\ref{model:period1} (with above mentioned parameters) has no nontrivial local conserved quantity, implying nonintegrability of the model.

\begin{figure}
    \centering
    \includegraphics{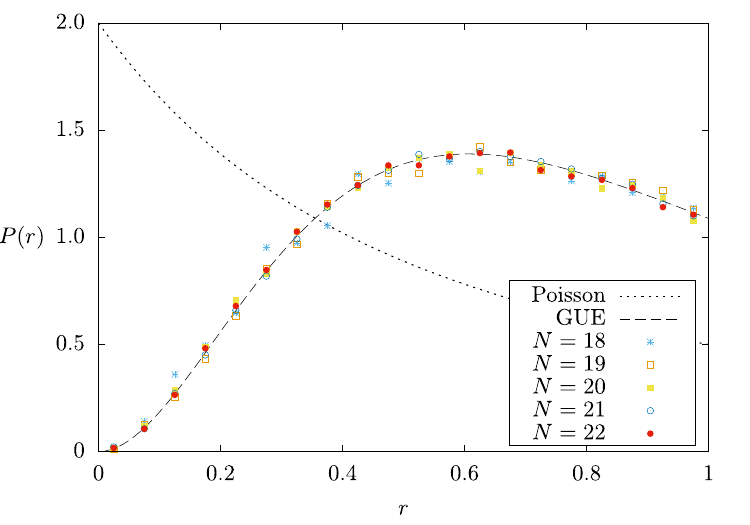}
    \caption{Distribution of the ratio of consecutive level spacings of Model~\ref{model:period1} in Theorem~\ref{theorem:model} in the main text. We use eigenenergies in the subspace of momentum $2\pi/N$.}
    \label{fig:LevelSpacing_model1}
\end{figure}

\section{Proof of Theorem~\ref{theorem:finite-temperature}}
\begin{proof}
Let $\vec{\sigma}=(\sigma_{1},\sigma_{2},\cdots,\sigma_{N})$ be a bit string of length $N$. Then, using $\vec{\sigma}$, we define the the computational basis as
\begin{align}
    \ket{\vec{\sigma}} &= \ket{\sigma_{1}} \otimes \ket{\sigma_{2}} \otimes \cdots \otimes \ket{\sigma_{N}}.
    \label{eq:ComputationalBasis}
\end{align}

As a preparation, we clarify the properties of the Hamiltonian satisfying the assumptions of the theorem. Since Pauli strings form an orthogonal basis of operators on the whole Hilbert space, the Hamiltonian $\hat{H}$ can be uniquely expressed as a linear combination of them:
\begin{align}
    \hat{H}
    &= \sum_{X(\subset\Lambda)} \sum_{\vec{\mu}\in\{x,y,z\}^{X}}
    J_{X}^{\vec{\mu}} \bigotimes_{j\in X} \hat{\sigma}_{j}^{\mu_{j}}.
    \label{eq:finite-temperature-Hamiltonian}
\end{align}
With this notation, $\hat{H}_\mathrm{int}$ in the main text is written as
\begin{align}
    \hat{H}_\mathrm{int}
    &= \sum_{\text{$X \cap L \neq \emptyset$ and $X \cap R \neq \emptyset$}} \sum_{\vec{\mu}\in\{x,y,z\}^{X}}
    J_{X}^{\vec{\mu}} \bigotimes_{j \in X} \hat{\sigma}_{j}^{\mu_{j}},
\end{align}
where $L=\{1,2,\cdots,N/2\}$ and $R=\{N/2+1,N/2+2,\cdots,N\}$. Since, $\hat{H}$ is translation invariant by assumption, we have
\begin{align}
    \hat{H} - \hat{H}_\mathrm{int}
    = \hat{H}_{N/2,\mathrm{OBC}} \otimes \hat{I}_R
    + \hat{I}_L \otimes \hat{H}_{N/2,\mathrm{OBC}}.
    \label{eq:H_int-H_OBC}
\end{align}
Here, $\hat{H}_{N/2,\mathrm{OBC}}$ is the Hamiltonian for the same system of length $N/2$, but with open boundary conditions rather than periodic boundary conditions:
\begin{align}
    \hat{H}_{N/2,\mathrm{OBC}}
    &= \sum_{X \subset L} \sum_{\vec{\mu}\in\{x,y,z\}^{X}}
    J_{X}^{\vec{\mu}} \bigotimes_{j \in X} \hat{\sigma}_{j}^{\mu_{j}}.
\end{align}
Under the complex conjugation with respect to the computational basis, the Pauli string behaves as
\begin{align}
    \bigotimes_{j \in X} \hat{\sigma}_{j}^{\mu_{j}}
    \longmapsto (-1)^{P_{X}^{\vec{\mu}}} \bigotimes_{j \in X} \hat{\sigma}_{j}^{\mu_{j}},
\end{align}
where $P_{X}^{\vec{\mu}} = |\{ \mu_j | j \in X, \mu_j = y \}|$ is the number of Pauli matrices along the $y$-direction, $\hat{\sigma}_j^y$, within the Pauli string. Thus, the complex conjugation transforms the Hamiltonian as
\begin{align}
    \hat{H}
    &= \sum_{X(\subset\Lambda)}
    \sum_{\substack{
        \vec{\mu}\in\{x,y,z\}^{X}\\
        \text{s.t. $P_{X}^{\vec{\mu}}$ is even}
    }}
    J_{X}^{\vec{\mu}} \bigotimes_{j \in X} \hat{\sigma}_{j}^{\mu_{j}}
    + \sum_{X(\subset\Lambda)}
    \sum_{\substack{
        \vec{\mu}\in\{x,y,z\}^{X}\\
        \text{s.t. $P_{X}^{\vec{\mu}}$ is odd}
    }}
    J_{X}^{\vec{\mu}} \bigotimes_{j \in X} \hat{\sigma}_{j}^{\mu_{j}} \nonumber\\
    \longmapsto \qquad
    \hat{H}^*
    &= \sum_{X(\subset\Lambda)}
    \sum_{\substack{
        \vec{\mu}\in\{x,y,z\}^{X}\\
        \text{s.t. $P_{X}^{\vec{\mu}}$ is even}
    }}
    J_{X}^{\vec{\mu}} \bigotimes_{j \in X} \hat{\sigma}_{j}^{\mu_{j}}
    - \sum_{X(\subset\Lambda)}
    \sum_{\substack{
        \vec{\mu}\in\{x,y,z\}^{X}\\
        \text{s.t. $P_{X}^{\vec{\mu}}$ is odd}
    }}
    J_{X}^{\vec{\mu}} \bigotimes_{j \in X} \hat{\sigma}_{j}^{\mu_{j}}.
\end{align}
Since the expansion in terms of Pauli strings is unique, for $\hat{H}$ to be a real matrix in the computational basis (i.e., $\hat{H}=\hat{H}^*$), $J_{X}^{\vec{\mu}}$ must be zero when $P_{X}^{\vec{\mu}}$ is odd. Hence, we obtain
\begin{align}
    \hat{H}_{N/2,\mathrm{OBC}}
    = \sum_{X \subset L}
    \sum_{\substack{
        \vec{\mu}\in\{x,y,z\}^{X}\\
        \text{s.t. $P_{X}^{\vec{\mu}}$ is even}
    }}
    J_{X}^{\vec{\mu}} \bigotimes_{j \in X} \hat{\sigma}_{j}^{\mu_{j}}.
\end{align}
Therefore, $\hat{H}_{N/2,\mathrm{OBC}}$ is also a real matrix in the computational basis.

We now proceed to prove Eq.~\eqref{eq:beta-state_Gibbs}. The EAP state $\ket{1;00}$ can be expanded in the computational basis for subsystems $L$ and $R$ as
\begin{align}
  \ket{1;00} \propto \sum_{\vec{\sigma}} \ket{\vec{\sigma}}_L \otimes \ket{\vec{\sigma}}_R.
\end{align}
Thus, using Eq.~\eqref{eq:H_int-H_OBC}, we have
\begin{align}
  \ket{\tilde{\beta}}
  &\propto \sum_{\vec{\sigma}}
  e^{- \frac{1}{4} \beta \hat{H}_{N/2,\mathrm{OBC}}} \ket{\vec{\sigma}}
  \otimes e^{- \frac{1}{4} \beta \hat{H}_{N/2,\mathrm{OBC}}} \ket{\vec{\sigma}} \nonumber\\
  &= \sum_{\vec{\sigma}}
  e^{- \frac{1}{4} \beta \hat{H}_{N/2,\mathrm{OBC}}} \ket{\vec{\sigma}}
  \otimes
  \left( \sum_{\vec{\sigma}'} {\ket{\vec{\sigma}'}\bra{\vec{\sigma}'}} \right)
  e^{- \frac{1}{4} \beta \hat{H}_{N/2,\mathrm{OBC}}} \ket{\vec{\sigma}} \nonumber\\
  &= \sum_{\vec{\sigma},\vec{\sigma}'}
  e^{- \frac{1}{4} \beta \hat{H}_{N/2,\mathrm{OBC}}} \ket{\vec{\sigma}}
  \otimes
  \braket{\vec{\sigma}'|e^{- \frac{1}{4} \beta \hat{H}_{N/2,\mathrm{OBC}}}|\vec{\sigma}} \ket{\vec{\sigma}'}.
  \label{eq:beta-tilde_purification}
\end{align}
Since $\hat{H}_{N/2,\mathrm{OBC}}$ is a real matrix with respect to the computational basis, it holds that
\begin{align}
  \braket{\vec{\sigma}|e^{- \frac{1}{4} \beta \hat{H}_{N/2,\mathrm{OBC}}}|\vec{\sigma}'}^*
  = \braket{\vec{\sigma}|e^{- \frac{1}{4} \beta \hat{H}_{N/2,\mathrm{OBC}}}|\vec{\sigma}'}
\end{align}
for any $\ket{\vec{\sigma}}$ and $\ket{\vec{\sigma}'}$.
Substituting this into Eq.~\eqref{eq:beta-tilde_purification}, we obtain
\begin{align}
  \ket{\tilde{\beta}}
  &\propto \sum_{\vec{\sigma},\vec{\sigma}'}
  e^{- \frac{1}{4} \beta \hat{H}_{N/2,\mathrm{OBC}}} \ket{\vec{\sigma}}\bra{\vec{\sigma}} e^{- \frac{1}{4} \beta \hat{H}_{N/2,\mathrm{OBC}}} \ket{\vec{\sigma}'}
  \otimes \ket{\vec{\sigma}'}
  = \sum_{\vec{\sigma}'}
  e^{-\frac{1}{2}\beta\hat{H}_{N/2,\mathrm{OBC}}} \ket{\vec{\sigma}'}
  \otimes \ket{\vec{\sigma}'}.
\end{align}
Therefore, for any observable $\hat{O}$ defined on the subsystem $L$, we get
\begin{align}
    \braket{\tilde{\beta}|\hat{O}|\tilde{\beta}}
    = \mathrm{Tr} [ \hat{\rho}_{N/2,\mathrm{OBC}}^\mathrm{can} \hat{O} ],
\end{align}
where $\hat{\rho}_{N/2,\mathrm{OBC}}^\mathrm{can}$ is the Gibbs state for $\hat{H}_{N/2,\mathrm{OBC}}$. Hence the thermodynamic limit yields
\begin{align}
    \lim_{N\to\infty} \braket{\tilde{\beta}|\hat{O}|\tilde{\beta}}
    = \lim_{N\to\infty} \mathrm{Tr} [ \hat{\rho}_{N/2,\mathrm{OBC}}^\mathrm{can} \hat{O} ].
\end{align}
In the thermodynamic limit, the Gibbs state converges to the KMS state regardless of whether periodic or open boundary conditions are imposed. Since we are now considering a one-dimensional system, there exists a unique KMS state at finite temperature~\cite{Araki1969,Araki1975}. Consequently, expectation values of local observables in the Gibbs state do not depend on boundary conditions in the thermodynamic limit. Thus, using Eq.~\eqref{eq:beta-tilde_beta}, we finally obtain
\begin{align}
    \lim_{N\to\infty} \braket{\beta|\hat{O}|\beta}
    = \lim_{N\to\infty} \mathrm{Tr} [ \hat{\rho}^\mathrm{can} \hat{O} ].
\end{align}
\end{proof}

\section{Degeneracy}
\subsection{Degeneracy in Model~\ref{model:period1}}
According to Theorem~\ref{theorem:model}, the EAP state $\ket{1;00}$ is an energy eigenstate with an eigenvalue $E=0$ of Model~\ref{model:period1} for arbitrary parameters. Without loss of generality, we can take $J^{zy}=0$ by an appropriate rotation around $y$ axis, so we set the parameters of Model~\ref{model:period1} as $J^{xy}=e,J^{yx}=1,J^{yz}=\pi,J^{zy}=0,h^{y}=\ln 7$. By the exact diagonalization, we find that $\ket{1;00}$ is doubly degenerate in the zero-momentum sector for $N=10,12,14,16$. Let us investigate entanglement properties of states in this eigenspace, which we will write as $\mathcal{H}_{k=0,E=0}$. Let $\ket{\perp}$ denote the state orthogonal to $\ket{1;00}$ in $\mathcal{H}_{k=0,E=0}$. All states in $\mathcal{H}_{k=0,E=0}$ can be expressed as a linear combination of $\ket{1;00}$ and $\ket{\perp}$:
\begin{align}
    \ket{\theta,\lambda}
    = \sqrt{1-\lambda} \ket{1;00} + e^{2 \pi i \theta} \sqrt{\lambda} \ket{\perp} \qquad (0 \leq \lambda \leq 1, 0 \leq \theta \leq 1).
    \label{eq:theta-lambda}
\end{align}

First, we investigate the bipartite entanglement in $\ket{\perp}$ (corresponding to the case of $\lambda=1$). We plot in Fig.~\ref{fig:bipartite-entanglement} the entanglement entropy $S_A$ between a subsystem $A=\{1,2,\cdots,\ell\}$ of length $\ell$ and its complement as a function of $\ell$. It can be seen that $\ket{\perp}$ is almost maximally entangled, but is different from EAP states.

Next, we confirm that all states in the eigenspace $\mathcal{H}_{k=0,E=0}$ are maximally entangled states. To investigate the coefficient of the volume-law scaling, we compute the entanglement entropy of $\ket{\theta,\lambda}$ between a subsystem of length $1$ and its complement for various $\lambda$ and $\theta$ and show in Fig.~\ref{fig:volume-low-coefficient} the deviation from the coefficient of the maximally entangled state, $1 - S_{A=\{1\}} / \log 2$. It can be observed that for all states in $\mathcal{H}_{k=0,E=0}$, the volume-law coefficients are significantly close to the maximal coefficient.

Thus, there are not any low entangled states in the eigenspace, and hence the EAP state is not a superposition of such states.

\begin{figure}
  \begin{minipage}[b]{0.58\columnwidth}
    \centering
    \includegraphics[width=\columnwidth]{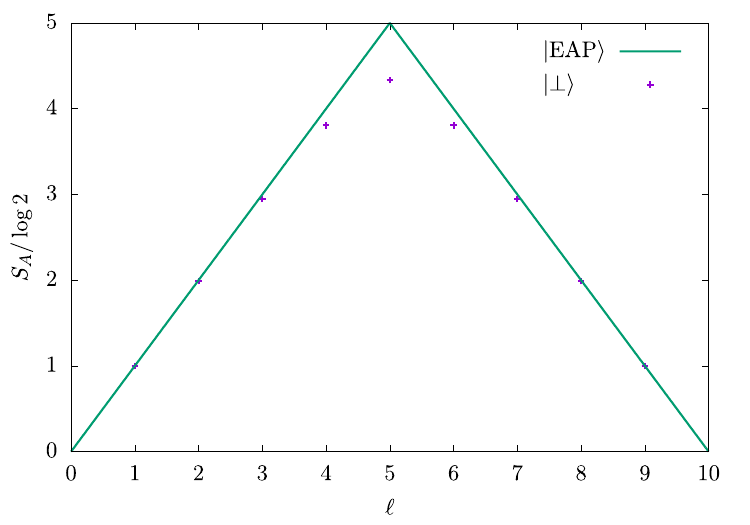}
    \caption{Entanglement entropy of the EAP state and its orthogonal degenerate state $\ket{\perp}$ between a subsystem $A=\{1,2,\cdots,\ell\}$ of length $\ell$ and its complement as a function of $\ell$ for Model~\ref{model:period1} with $N=10$. We set the parameters as $J^{xy}=e,J^{yx}=1,J^{yz}=\pi,J^{zy}=0,h^{y}=\ln 7$.}
    \label{fig:bipartite-entanglement}
  \end{minipage}
  \hfill
  \begin{minipage}[b]{0.38\columnwidth}
    \centering
    \includegraphics[width=\columnwidth]{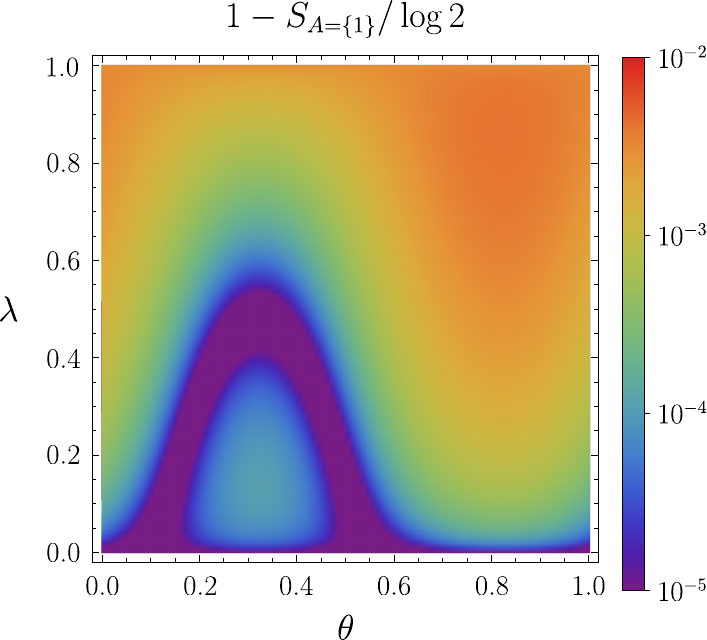}
    \caption{Deviation of the volume-law coefficient of the entanglement entropy of zero-energy eigenstates $\ket{\theta,\lambda}$ defined by Eq.~\eqref{eq:theta-lambda} in the zero-momentum sector of Model~\ref{model:period1} from that of the maximally entangled state. We set the parameters as $J^{xy}=e,J^{yx}=1,J^{yz}=\pi,J^{zy}=0,h^{y}=\ln 7$ and $N=10$.}
    \label{fig:volume-low-coefficient}
  \end{minipage}
\end{figure}

\subsection{Nondegenerate Hamiltonian with next-nearest-neighbor interactions}
In this subsection, by extending the Hamiltonian~(\ref{eq:H_period2_11}) to the next-nearest-neighbor interacting one, we provide a Hamiltonian having an EAP state as an eigenstate that is nondegenerate in the corresponding momentum sector.

Suppose that $\hat{H}$ is translation invariant and satisfies $J_{X}^{\vec{\mu}}=0$ for any subset $X$ with $D(X)\ge 4$. Then, it can be characterized by $48$ coupling constants, $J^{\mu\lambda\nu}$, $J^{\mu\nu}$ and $h^{\mu}$ ($\mu,\nu=x,y,z$ and $\lambda=x,y,z,0$, where $\hat{\sigma}^{0}:=\hat{1}$). From Theorem~\ref{theorem:eigenstate-condition} of the main text, it is straightforward to show that, when $N/2$ is a multiple of $2$, the EAP states $\ket{2;11,10}$ and $\ket{2;10,11}$ are eigenstates of $\hat{H}$ if and only if $\hat{H}$ can be written as
\begin{align}
    \hat{H} &= \sum_{j=1}^{N} \bigl(
        J^{xzx} \hat{\sigma}_{j}^{x} \hat{\sigma}_{j+1}^{z} \hat{\sigma}_{j+2}^{x}
        + J^{yzy} \hat{\sigma}_{j}^{y} \hat{\sigma}_{j+1}^{z} \hat{\sigma}_{j+2}^{y}
        + J^{xzy} \hat{\sigma}_{j}^{x} \hat{\sigma}_{j+1}^{z} \hat{\sigma}_{j+2}^{y}
        + J^{yzx} \hat{\sigma}_{j}^{y} \hat{\sigma}_{j+1}^{z} \hat{\sigma}_{j+2}^{x}
                \nonumber\\
        &\hspace{36pt}+ J^{zzz} \hat{\sigma}_{j}^{z} \hat{\sigma}_{j+1}^{z} \hat{\sigma}_{j+2}^{z}
        +J^{xx} \hat{\sigma}_{j}^{x} \hat{\sigma}_{j+1}^{x}
        + J^{yy} \hat{\sigma}_{j}^{y} \hat{\sigma}_{j+1}^{y}
        + J^{xy} \hat{\sigma}_{j}^{x} \hat{\sigma}_{j+1}^{y}
        + J^{yx} \hat{\sigma}_{j}^{y} \hat{\sigma}_{j+1}^{x}
        + h^{z} \hat{\sigma}_{j}^{z}
    \bigr).
    \label{eq:H_period2_11_NNN}
\end{align}
Here all parameters are arbitrary, and hence it is an extension of Eq.~(\ref{eq:H_period2_11}).

Because the EAP states $\ket{2;11,10}$ and $\ket{2;10,11}$ are related to each other by translation $\hat{\mathcal{T}}$ as
\begin{align}
    \hat{\mathcal{T}}\ket{2;11,10}
    &=+\ket{2;10,11},
    \label{eq:EAP_period2_11_Translation1}\\
    \hat{\mathcal{T}}\ket{2;10,11}
    &=-\ket{2;11,10},
    \label{eq:EAP_period2_11_Translation2}
\end{align}
their superposition $(\ket{2;11,10}\pm i\ket{2;10,11})/\sqrt{2}$ is included in the eigenspace of translation with the momentum $k=\pm \pi/2$. Therefore, we investigate degeneracy of energy eigenvalues in the subspace of $k=\pi/2$. We set $h^z=1$, $J^{xx}=e$, $J^{xy}=\pi$, $J^{yx}=\ln 2$, $J^{yy}=\ln 3$, $J^{xzx}=\ln 5$, $J^{xzy}=\ln 7$, $J^{yzx}=\ln 11$, $J^{yzy}=\ln 13$, and $J^{zzz}=\ln 17$. By exact diagonalization, we numerically find that, at least for $N=8,12,16,20$, the eigenvalue $E=0$ is nondegenerate in the subspace of $k=\pi/2$.

We also verify the nonintegrability of model~(\ref{eq:H_period2_11_NNN}). Because of the existence of $J^{zzz}$, model~(\ref{eq:H_period2_11_NNN}) is not mapped to a free fermionic (integrable) system by the Jordan-Wigner transformation. We confirm nonintegrability of model~(\ref{eq:H_period2_11_NNN}) by calculating the distribution of the ratio of consecutive level spacings, as in Fig.~\ref{fig:LevelSpacing_model1}. Figure~\ref{fig:LevelSpacing_NNN} plots the distribution constructed from energy eigenvalues with eigenmomentum $k=\pi/2$ and parity $\mathcal{P}=\pm 1$. (Here the eigenmomentum $k$ is related to the eigenvalue of translation $\lambda$ as $\lambda=e^{-ik}$, and the parity $\mathcal{P}$ is the eigenvalue of $\mathbb{Z}_{2}$ symmetry $\bigotimes_{j=1}^{N}\hat{\sigma}_{j}^{z}$.) This plot is well described by the Gaussian unitary ensemble distribution~\cite{Atas2013} (dashed line) and well separated from the Poisson distribution~\cite{Atas2013} (dotted line), indicating the nonintegrability of the model.

\begin{figure}
    \centering
    \includegraphics{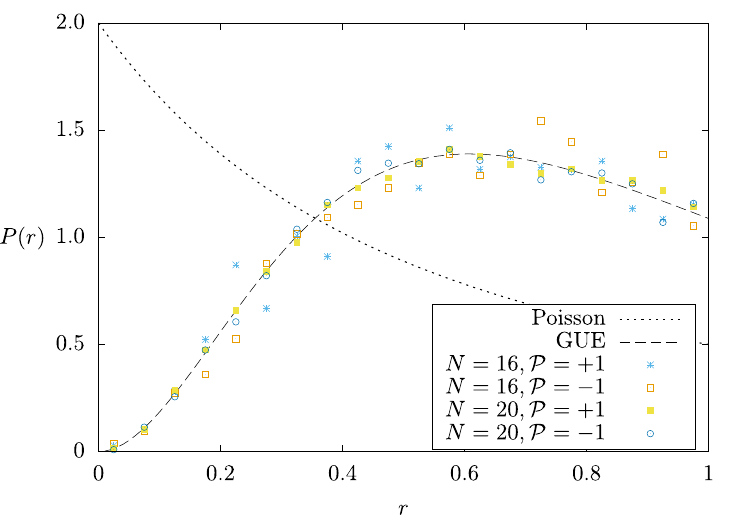}
    \caption{Distribution of the ratio of consecutive level spacings of model~(\ref{eq:H_period2_11_NNN}). We use eigenenergies in the subspace of momentum $k=\pi/2$ and parity $\mathcal{P}=\pm 1$ (regarding rotation by $\pi$ around $z$-axis). The parameters are given below Eq.~(\ref{eq:EAP_period2_11_Translation2}).}
    \label{fig:LevelSpacing_NNN}
\end{figure}

Note that, because the interference term between $\ket{2;11,10}$ and $\ket{2;10,11}$ does not affect the expectation values of local observables whose support size are less than $N/2$, any state described by a linear combination of $\ket{2;11,10}$ and $\ket{2;10,11}$ is locally indistinguishable from the maximally mixed state.

Combining all results of this subsection, we can say that the state $(\ket{2;11,10}\pm i\ket{2;10,11})/\sqrt{2}$ is a thermal eigenstate of the nonintegrable Hamiltonian~(\ref{eq:H_period2_11_NNN}), and is nondegenerate in the corresponding momentum sector~\footnote{The results of exact diagonalization also show that, at least for $N=8,16,20$, degeneracy at $E=0$ in the whole Hilbert space remains only two. This means that all eigenstates of the Hamiltonian~(\ref{eq:H_period2_11_NNN}) with eigenenergy $E=0$ are given by (linear combinations of) the EAP states $\ket{2;11,10}, \ket{2;10,11}$.}.

\section{Relation to previously constructed thermal eigenstates}
In this section, we discuss the relation between our Theorem~\ref{theorem:eigenstate-condition} and previously constructed thermal eigenstates of nonintegrable models~\cite{Udupa2023,Ivanov2024}.

\subsection{Relation to the results by A. Udupa, S. Sur, \textit{et al.}}

Now we discuss relation to the results by A. Udupa, S. Sur, \textit{et al.}~\cite{Udupa2023}.
They investigated the model described by the Hamiltonian
\begin{align}
    \hat{H}_{ZZZ}=\sum_{j=1}^{N}\bigl(J\hat{\sigma}_{j}^z\hat{\sigma}_{j+1}^z\hat{\sigma}_{j+2}^z+h\hat{\sigma}_{j}^x\bigr)
\end{align}
with the periodic boundary condition.
This model is expected to be nonintegrable because its level spacing distribution is described by the Wigner-Dyson distribution~\cite{Udupa2023}.

They showed that $\hat{H}_{ZZZ}$ has an EAP state $\ket{1;11}$ as an eigenstate with the eigenvalue $0$.
This fact can be readily verified by checking Eq.~\eqref{eq:eigenstate-condition} and using (\ref{statement:eigenstate-condition_3}) $\Rightarrow$ (\ref{statement:eigenstate-condition_2}) of our Theorem~\ref{theorem:eigenstate-condition}.

Note that $\hat{H}_{ZZZ}$ has exponentially large degeneracy, which is greater than or equal to $2^{N/2}$ as they showed, at the eigenvalue $0$.
They constructed some degenerate eigenstates other than $\ket{1;11}$, but not all of them.
Thus, their results do not exclude the possibility that, for another basis of zero energy eigenspace, each of the basis vectors becomes hardly entangled.
(By contrast, our Fig.~\ref{fig:volume-low-coefficient} shows that any superpositions of degenerate states have volume law entanglement with coefficients almost the same as the maximal value $\log 2$.)

Note also that, although they argued that constructed eigenstates including $\ket{1;11}$ are quantum many-body scar states, 
the EAP state $\ket{1;11}$ is thermal in the most fundamental sense as explained in the main text.
(We agree that some other degenerate eigenstates that contain singlet pairs between neighboring sites~\cite{Udupa2023} are many-body scar states because they can be distinguished from the Gibbs state at $\beta=0$ by a local observable on the neighboring sites.)

\subsection{Relation to the results by A. N. Ivanov and O. I. Motrunich}

Next, we discuss the relation to the results by A. N. Ivanov and O. I. Motrunich~\cite{Ivanov2024}.
They investigated the PXP model, which is described by the Hamiltonian
\begin{align}
    \hat{H}_{PXP}=\sum_{j=1}^{N}\hat{P}_{j-1}\hat{\sigma}_{j}^x\hat{P}_{j+1}
\end{align}
with the periodic boundary condition, where $\hat{P}_{j}=(1+\hat{\sigma}_j^{z})/2$.
Let $\hat{\mathcal{P}}_{\mathrm{Ryd}}$ be the projection operator to the nearest neighbor Rydberg blockaded subspace, which is spanned by the computational basis vectors $\ket{\vec{\sigma}}$ whose bitstrings $\vec{\sigma}$ do not contain neighboring ``$11$.''
[Here, the computational basis is defined by Eq.~\eqref{eq:ComputationalBasis}.]
This projection operator can be written as
\begin{align}
    \hat{\mathcal{P}}_{\mathrm{Ryd}}=\prod_{j=1}^{N}(1-\hat{Q}_{j}\hat{Q}_{j+1}),
\end{align}
where $\hat{Q}_{j}=1-\hat{P}_{j}$ and the periodic boundary condition is imposed again.

They showed that 
the projection of an EAP state, $\hat{\mathcal{P}}_{\mathrm{Ryd}}\ket{1;01}$, is an eigenstate of $\hat{H}_{PXP}$ (up to normalization) with the eigenvalue $0$.
This fact can be verified from our Theorem~\ref{theorem:eigenstate-condition} as follows. 
By checking Eq.~\eqref{eq:eigenstate-condition} and using (\ref{statement:eigenstate-condition_3}) $\Rightarrow$ (\ref{statement:eigenstate-condition_2}) of our Theorem~\ref{theorem:eigenstate-condition}, we can show that the EAP state $\ket{1;01}$ itself satisfies $\hat{H}_{PXP}\ket{1;01}=0$. 
Furthermore, we can show that each local term $\hat{P}_{j-1}\hat{\sigma}_{j}^x\hat{P}_{j+1}$ in $\hat{H}_{PXP}$ commutes with $\hat{\mathcal{P}}_{\mathrm{Ryd}}$ by naively calculating the commutator, and hence $\hat{H}_{PXP}$ also commutes with $\hat{\mathcal{P}}_{\mathrm{Ryd}}$. Thus we have
\begin{align}
    \hat{H}_{PXP}\hat{\mathcal{P}}_{\mathrm{Ryd}}\ket{1;01}
    =\hat{\mathcal{P}}_{\mathrm{Ryd}}\hat{H}_{PXP}\ket{1;01}
    =0,
\end{align}
which corresponds to the result explained above.

In addition, we show in the following that the state $\hat{\mathcal{P}}_{\mathrm{Ryd}}\ket{1;01}$ is thermal in the sense that it is indistinguishable from the maximally mixed state of the Rydberg blockaded subspace, whose density matrix is given by $\hat{\mathcal{P}}_{\mathrm{Ryd}}$ up to normalization.
We examine the reduced density matrix on consecutive $m$ sites $\{1,...,m\}$ ($m\le N/2$) for both the state $\hat{\mathcal{P}}_{\mathrm{Ryd}}\ket{1;01}$ and the maximally mixed state, and compare them.
Let $\mathcal{F}_{m}^{(\mathrm{o})}$ be the set of bitstrings that satisfy the Rydberg blockaded condition for the \emph{open} boundary condition.
As shown in Sec.~II of Supplemental Material of Ref.~\cite{Ivanov2024}, the reduced density matrix on sites $\{1,...,m\}$ ($m\le N/2$) for the state $\hat{\mathcal{P}}_{\mathrm{Ryd}}\ket{1;01}$ is given by (up to normalization)
\begin{align}
    \sum_{\vec{\sigma}\in\mathcal{F}_{m}^{(\mathrm{o})}}F_{\frac{N}{2}+2-m-\sigma_1-\sigma_m}\ket{\vec{\sigma}}\bra{\vec{\sigma}},
    \label{eq:RDM_ProjectedEAP}
\end{align}
where $F_n$ is the $n$-th Fibonacci number.
On the other hand, by a calculation similar to theirs, we can show that 
the reduced density matrix on the same sites for the maximally mixed state is given by (up to normalization)
\begin{align}
    \sum_{\vec{\sigma}\in\mathcal{F}_{m}^{(\mathrm{o})}}F_{N+2-m-\sigma_1-\sigma_m}\ket{\vec{\sigma}}\bra{\vec{\sigma}}.
    \label{eq:RDM_ProjectedMaxMix}
\end{align}
Hence, in the limit of $N\to\infty$ (with fixed $m$), Eqs.~\eqref{eq:RDM_ProjectedEAP} and \eqref{eq:RDM_ProjectedMaxMix} give the same reduced density matrices. 
This indicates that $\hat{\mathcal{P}}_{\mathrm{Ryd}}\ket{1;01}$ is indistinguishable from the maximally mixed state by the expectation values of any local observables.

Thus, the state $\hat{\mathcal{P}}_{\mathrm{Ryd}}\ket{1;01}$ studied in Ref.~\cite{Ivanov2024} is also an example of exact thermal eigenstates of nonintegrable systems (because the PXP model is known to be nonintegrable~\cite{Park2024}),
although they referred to the state $\hat{\mathcal{P}}_{\mathrm{Ryd}}\ket{1;01}$ as a quantum many-body scar state.
As explained in the main text, our terminology is based on one of the standard notions of thermal equilibrium, ``MITE~\cite{Mori2018}.'' 
\end{document}